\documentclass[11p, reqno]{amsart}
\usepackage{amsmath}
\usepackage{amssymb}
\numberwithin{equation}{section}
\usepackage{amsthm}
\usepackage{braket}
\pdfoutput=1
\usepackage{appendix} 
\usepackage[utf8]{inputenc}
\usepackage[T1]{fontenc}
\usepackage{microtype}

\usepackage{fourier}
\usepackage{caption}
\usepackage{mathrsfs}
\usepackage[colorlinks=true, pdfstartview=FitV, linkcolor=blue, citecolor=blue, urlcolor=blue]{hyperref}
\theoremstyle{plain}
\newtheorem{proposition}{Proposition}[section]
\newtheorem{corollary}[proposition]{Corollary}

\newtheorem{lemma}[proposition]{Lemma}
\newtheorem{theorem}[proposition]{Theorem}
\theoremstyle{definition}
\newtheorem{definition}[proposition]{Definition}
\newtheorem{example}[proposition]{Example}
\newtheorem{remark}[proposition]{Remark}

\usepackage[usenames,dvipsnames]{xcolor}
\usepackage{url}
\usepackage{tikz}
\usepackage{pgfplots}
\pgfplotsset{compat=newest,ticks=none}
\usetikzlibrary{arrows,calc}
\usepackage{verbatim}
\usetikzlibrary{%
    decorations.pathreplacing,%
    decorations.pathmorphing%
}

\DeclareMathOperator{\sgn}{sgn}

\newcommand{\R}{\mathbf{R}}
\newcommand{\C}{\mathbf{C}}
\newcommand{\D}{\mathcal{D}}

\newcommand{\hi}{\widehat{\Phi}}

\newcommand{\abs}[1]{\left\lvert #1 \right\rvert}
\newcommand{\avg}[1]{\bigl\langle #1 \bigr\rangle}
\newcommand{\hoc}[1]{#1^+}

\newcommand{\jump}[1]{\bigl[ #1 \bigr]}
\newcommand{\floor}[1]{\left\lfloor #1 \right\rfloor}

\numberwithin{equation}{section}
\usepackage{mathscinet}
\usepackage{cite}

\theoremstyle{plain}

\usepackage{tikz}
\tikzset{
  source/.style={circle,draw=black!100,fill=black!50,inner sep = 0,minimum size=2mm},
  sink/.style={circle,draw=black!100,fill=white,inner sep = 0,minimum size=2mm}
}

\usepackage{microtype}
\begin{document}

\title[Peakon Problem for the modified Camassa-Holm equation]{Lax integrability and the peakon problem for the modified Camassa-Holm equation}

\author{Xiangke Chang}
\address{LSEC, Institute of Computational Mathematics and Scientific Engineering Computing, AMSS, Chinese Academy of Sciences, P.O.Box 2719, Beijing 100190, PR China, and the Department of Mathematics and Statistics, University of Saskatchewan, 106 Wiggins Road, Saskatoon, Saskatchewan, S7N 5E6, Canada.}
\email{changxk@lsec.cc.ac.cn}
\thanks{The first author was supported in part by the Natural  Sciences and Engineering Research Council of Canada (NSERC), the Department of Mathematics and Statistics of the University of Saskatchewan, PIMS postdoctoral fellowship and 
the LSEC,Institute of Computational Mathematics and Scientific Engineering Computing, AMSS, CAS}
  \author{Jacek Szmigielski}
  \address{Department of Mathematics and Statistics, University of Saskatchewan, 106 Wiggins Road, Saskatoon, Saskatchewan, S7N 5E6, Canada.}
\email{szmigiel@math.usask.ca}
\thanks{The second author was supported in part by NSERC \#163953.}


\maketitle

\begin{abstract}
\textit{Peakons} are special weak solutions of a class of nonlinear 
partial differential equations modelling non-linear phenomena such as the breakdown of regularity and the onset of shocks. We show that the natural concept of weak solutions in the case of the modified Camassa-Holm equation studied in this paper is dictated by  the distributional compatibility of its 
Lax pair and, as a result, it differs from the one proposed and used in the literature 
based on the concept of weak solutions used for equations of the Burgers type.  
Subsequently, we give a complete construction of peakon solutions satisfying 
the modified Camassa-Holm equation in the sense of distributions; our approach is based on 
solving certain inverse boundary value problem the solution of which hinges on 
a combination of classical techniques of analysis involving 
Stieltjes' continued fractions and multi-point 
  Pad\'{e} approximations. We propose sufficient conditions 
  needed to ensure the global existence of peakon solutions and 
  analyze the large time asymptotic behaviour whose special features
  include a formation of pairs of peakons which share 
  asymptotic speeds, as well as \textit{Toda-like} sorting property.

\end{abstract}
Pad{\'e} approximation.






\tableofcontents

\section{Introduction}

\vspace{0.5 cm}  
The nonlinear partial differential equation
\begin{equation}\label{eq:m1CH}
m_t+\left((u^2-u_x^2) m\right)_x=0,  \qquad 
m=u-u_{xx},
\end{equation}
is an intriguing modification of the Camassa-Holm equation (CH) \cite{CH}: 
\begin{equation} \label{eq:CH}
m_t+u m_x +2u_x m=0, \qquad  m=u-u_{xx}, 
\end{equation} 
for the shallow water waves.  
Originally, equation \eqref{eq:m1CH} appeared in the papers of Fokas \cite{fokas1995korteweg}, Fuchssteiner \cite{fuchssteiner1996some},  Olver and Rosenau\cite{olver1996tri} and was, later, rediscovered by Qiao  
 \cite{qiao2006new,qiao2007new}.  

  We note that the derivation of this equation in \cite{olver1996tri} followed from the general method of tri-Hamiltonian duality applied to the bi-Hamiltonian representation of the modified Korteweg-de Vries equation (see also 
  \cite{kang-liu-olver-qu} for a recent generalization of this idea). Since the CH equation can be obtained from the Korteweg-de Vries equation
by the same tri-Hamiltonian duality, it is therefore natural to refer to equation \eqref{eq:m1CH} as the modified CH equation (mCH), in full agreement with other authors \cite{gui2013wave,liu2014orbital}, even though the name FORQ to denote 
\eqref{eq:m1CH} is sometimes used as well (e.g. \cite{Himonas4peak}, \cite{himonas2014cauchy}).  

We are interested in the class of non-smooth solutions of \eqref{eq:m1CH} given by the \textit{peakon ansatz} \cite{CH,qiao2012integrable,gui2013wave}, that is, we assume 
\begin{equation} \label{eq:peakonansatz}
u=\sum_{j=1}^n m_j (t)e^{-\abs{x-x_j(t)}}, \, 
\end{equation} 
where all coefficients $m_j(t)$ are taken to be positive, 
and hence 
\begin{equation*}
~m=u-u_{xx}=2\sum_{j=1}^n m_j \delta_{x_j}
\end{equation*} 
is a positive discrete measure.   The relevance of this 
ansatz proved to be supported by the fact that these special solutions 
seem to capture main attributes of solutions of this class of equations: the breakdown of regularity which can be interpreted as collisions of peakons, and the nature of 
long time asymptotics which can be loosely described as peakons becoming 
free particles in the asymptotic region \cite{bss-moment}.  For the CH equation peakons do not 
exhibit any asymptotic cooperative behaviour, while for other equations, for example 
for the Geng-Xue equation \cite{geng-xue:cubic-nonlinearity} or the Novikov equation \cite{novikov:generalizations-of-CH},  one observes 
pairing or even more elaborate patterns of clustering of peakons in the asymptotic 
region \cite{kardell-PhD, kardell-lundmark}.  At the same time, for still not entirely clear reasons, the 
peakon dynamics has an ever growing number of connections with 
classical analysis.  This was observed for the first time 
in the CH case \cite{bss-stieltjes} where the dynamics of peakons was shown to be 
related, in fact, solved, in terms of the classical theory of Stieltjes continued fractions - the connection that goes through the fundamental theory of the inhomogeneous string of M.G. Krein \cite{dym-mckean-gaussian} - eventually leading to sharp estimates on the patterns 
of the breakdown of regularity \cite{bss-moment, McKean-Asianbreakdown, McKean-breakdown}.  

Although it does not seem possible in a short introduction to 
give justice to the enormous literature on the CH equation we would like to 
mention a few works related to the issues raised in the present paper.  Thus 
\cite{constantin-escher} discusses the concept of weak solutions 
for the CH equation that set the stage for numerous studies of 
related equations as well as it gave the first general results regarding the wave breaking for CH.  In 
\cite{constantin-strauss} the authors discuss the issue of stability of CH peakons, 
while the stability of multipeakons is discussed in \cite{molinet-n-peakons, molinet-multipeakons}.  There is also a considerable literature on 
the use of peakons in designing numerical schemes (see e.g. \cite{holden-numer-peakons})  and the issue of continuing solutions past the 
breakdown of regularity (collisions of peakons) \cite{holden-dissip-peakons}.  

Meanwhile the literature on the peakon ansatz has grown considerably since its discovery  in \cite{CH}.
In the following years the peakon ansatz was successfully applied to another, well studied by now,  equation, namely 
the Degasperis-Procesi equation \cite{dp}
\begin{equation} \label{eq:DP}
m_t+u m_x +3u_x m=0, \qquad  m=u-u_{xx}, 
\end{equation} 
which despite its superficial similarity to the CH equation \eqref{eq:CH} has in addition 
shock solutions \cite{coclite-karlsen-DPwellposedness, coclite-karlsen-DPuniqueness, lundmark-shockpeakons}, while its peakon sector leads to new questions 
regarding Nikishin systems \cite{ns}  studied in approximation theory 
\cite{ls-cubicstring, bertola-gekhtman-szmigielski:cauchy}.  For 
potential applicability to water wave theory the reader is 
invited to consult \cite{Constantin-Lannes}; for a discussion of weak solutions see \cite{liu-escher}; \cite{liu1, liu2} present important results regarding stability, and finally \cite{sz1, sz2} deal with collisions of peakons and the onset of shocks in the form of shockpeakons\cite{lundmark-shockpeakons}.  

Another feature of peakon sectors of Lax integrable 
peakon equations is the \textit{omnipresence} of \textit{total positivity}\cite{gantmacher-krein, karlin-TP}.  In its simplest 
form, namely speaking of matrices, a totally positive matrix is a matrix whose minors, of all sizes, are positive.  This concept is then generalized to 
kernels of linear integral equations.  
Total positivity appears  in all peakon 
problems known to us, although admittedly we cannot yet explain from first principles the underlying reasons for the presence of such a strong form of 
positivity; however we remark that 
peakons are in a nutshell disguised \textit{oscillatory systems} in the sense 
of Gantmacher and Krein \cite{gantmacher-krein}.

What is germane to this paper is that the peakon problem at hand is 
coming from studying a distibutional Lax pair which forces us to view \eqref{eq:m1CH} as a distribution equation, requiring in particular that we 
define the product $u_x^2 m$.  With this in mind we show in Appendix \ref{lax_mch} that the 
choice consistent with Lax integrability is to take $u_x^2 m$ to mean $\langle u_x^2 \rangle m$, 
where $\langle f \rangle$ denotes the average function (the arithmetic average 
of the right hand and left hand limits). Subsequently, equation \eqref{eq:m1CH} 
reduces to the system of ODEs: 
\begin{equation}\label{eq:xmODE}
\dot m_j=0, \qquad \dot x_j=u(x_j)^2-\avg{u_x^2 }(x_j),  
\end{equation}
or, more explicitly, assuming the ordering condition $x_1< x_2<\cdots<x_n$,
\begin{equation}\label{mCH_ode}
\dot{m}_j=0, \qquad 
\dot{x}_j=2\sum_{\substack{1\leq k\leq n,\\k\neq j}}m_jm_ke^{-|x_j-x_k|}+4\sum_{1\leq i<j<k\leq n}m_im_ke^{-|x_i-x_k|}.  
\end{equation}
In broad terms we can say that our general interest in \eqref{eq:m1CH} is to understand how integrability manifests itself in the non-smooth sector of solutions, in particular how it determines the properties of, initially, ill-defined operations, 
which acquire well-defined meaning thanks to the condition of Lax integrability.  

We note that the system given by \eqref{mCH_ode} is not the same as the one proposed in \cite{gui2013wave}; 
the difference being precisely in the definition of the singular product $u_x^2 m$. 
We clarify the details of the difference in the remark below.  
\begin{remark}\label{re:olver}
  In \cite{gui2013wave}, Gui, Liu, Olver  and Qu showed that the mCH equation admits weak n-peakon solutions with $x_j,m_j$ satisfying
  \begin{equation}\label{mch_ode_GLQ}
   \dot{m}_j=0, \qquad 
  \dot{x}_j=\frac{2}{3}m_j^2+2\sum_{\substack{1\leq k\leq n,\\k\neq j}}m_jm_ke^{-|x_j-x_k|}+4\sum_{1\leq i<j<k\leq n}m_im_ke^{-|x_i-x_k|},  
    \end{equation}
 (these equations also appear as a special case in \cite{qiao2012integrable}); we note that these equations differ from \eqref{mCH_ode} by the constant term $\frac{2}{3}m_j^2$. For identical $m_j$ this term can be absorbed by redefining $x_j$ but in general 
 this cannot be done without violating the invariance of $|x_i-x_k|$.  
 It is not difficult to verify that, following the definition of weak solutions 
 adopted in \cite{gui2013wave}, 
 the singular product $u_x^2m$ appearing in \eqref{eq:m1CH} equals to 
 \begin{equation*}
 \big(\frac{\avg{u_x^2}+2\avg{u_x}^2}{3}\big) m, 
 \end{equation*} 
 which is an abbreviated way of saying that the value of the multiplier 
 of $\delta_{x_j}$ equals $\frac{\avg{u_x^2}+2\avg{u_x}^2}{3}(x_j)$.  This 
 is markedly different than what the Lax integrability implies for the multiplier, namely $\avg{u_x^2}(x_j)$.  
 Indeed, in our case, as we shall prove in Appendix \ref{lax_mch}, 
 \eqref{mCH_ode}  can be derived from the compatibility condition of a distribution  Lax pair, which in turn leads to explicit solutions of these equations by the 
 inverse spectral method, following a successful solution to the appropriate inverse problem.  Finally, in view of the comments above our solution is 
 also a solution to the special case of the peakon problem in \cite{gui2013wave}
 for which all the masses $m_j$ are assumed equal.

 For other work related to \eqref{eq:m1CH} done recently the reader is invited to consult \cite{liu-liu-qu, liu2014orbital, bies-gorka-reyes}.  
\end{remark} 
 \begin{remark} 
 Another important feature that sets apart our definition of peakons is that 
 the Sobolev $H^1$ norm of $u$  defined by \eqref{eq:peakonansatz} for peakons satisfying 
 equations \eqref{mCH_ode}  is preserved.  In other words
 \begin{equation*}
\frac{d}{dt}||u||_{H^1}=0. 
\end{equation*}
Even though this point is fully explained in the followup shorter paper 
\cite{chang-szmigielski-short1mCH}, we nevertheless compute $||u||^2_{H^1}$ in Corollary \ref{cor:sH1norm} 
 in terms of spectral variables, obtaining trace-like identity 
akin to the one known from the CH theory \cite{McKean-Fred}.  We stress 
that the time preservation of the $H^1$ norm is of considerable importance 
if one recalls that  one of the Hamiltonians 
defining the theory of \eqref{eq:m1CH} is 
$\mathcal{H}_1=||u||^2_{H^1}$ (see, for example, \cite{kang-liu-olver-qu}).  
\end{remark}

In the present paper, we shall formulate and apply an inverse spectral method to solve the peakon ODEs 
\eqref{mCH_ode} and hence \eqref{eq:m1CH} under the following assumptions: 
\begin{enumerate}
\item all $m_k$ are positive, 
\item the initial positions are assumed to be ordered as $x_1(0)<x_2(0)<\cdots<x_n(0)$. 
\end{enumerate} 

\begin{remark}
  If $m_k$ are negative, the corresponding problem may be solved by the transformation $m_k\rightarrow -m_k$. This results in pure antipeakon solutions.
\end{remark}
In the remainder of this introduction we outline the content of 
individual sections, highlighting the main results.  
Thus in Section \ref{sec:Lax} we reformulate the Lax pair in a way suitable 
for further analysis; in particular, for the peakon ansatz we obtain 
a difference equation and we solve explicitly the affiliated initial 
value problem.  This section uses in an essential way the result from Appendix 
\ref{lax_mch} about the admissible ways of defining the distributional 
Lax pair.  

In Section \ref{sec:FSM} we give a full characterization of 
the spectrum of the boundary value problem from Section \ref{sec:Lax}.  
We prove that the spectrum is positive 
and simple, and in this sense the mCH peakons confirm the 
``experimental" fact that all known integrable peakon equations have a substantial amount of positivity built in. 
The spectral data, which involves not only the eigenvalues but also 
some positive constants known in scattering theory as the \textit{norming 
constants}, are elegantly encoded in the Weyl function $W(z)$ of the boundary 
value problem and the main theorem, namely Theorem \ref{thm:W}, which 
states that $W(z)$ is a shifted Stieltjes transform is proven in its 
entirety therein.   

In Section \ref{sec:IP} we solve the inverse boundary value problem 
which in a nutshell amounts to reconstructing the measures $g, h$ appearing 
in the original formulation of the boundary value problem \eqref{eq:xLaxBVP}
from the spectral data encoded in the Weyl function $W(z)$.  We subsequently give 
two constructions of the inverse map: one is based on recurrence relations and 
Stieltjes' method of continued fractions, the other method is 
explicit and it involves certain Cauchy-Jacobi interpolation problem
which is shown in Theorem \ref{thm:detsolISP} to admit an explicit solution in terms of Cauchy-Stieltjes-Vandermonde matrices introduced in 
Definition \ref{def:CSV}.   

In Sections \ref{sec:evenpeakons} and \ref{sec:oddpeakons} we analyze 
the actual peakon solutions $u$ constructed out of 
the peakon ansatz \eqref{eq:peakonansatz} and the determintal solution of 
the inverse problem studied in Section \ref{sec:IP}.  This material is covered in two 
sections because there are some subtle differences in the character of solutions 
depending on whether the total number of peakons, $n$, is even or odd.  
In either case we present and prove sufficient conditions for the 
global existence of peakon solutions.  This is done in Theorems 
\ref{thm:global} and \ref{thm:global(odd)}.  We also give 
large time asymptotic formulas for peakons, showing that 
in both cases the peakons form asymptotic pairs.

\section{The Lax formalism: the boundary value problem} \label{sec:Lax}
The Lax 
pair for \eqref{eq:m1CH} reads \cite{qiao2006new}: 
\begin{equation} \label{eq:xtLax}
\Psi_x=\frac12U \Psi, \quad  \Psi _t =\frac12 V \Psi, \quad  \Psi=\begin{bmatrix} \Psi_1\\\Psi_2 \end{bmatrix} 
\end{equation} 
with 
\begin{equation*} 
U=\begin{bmatrix} -1 &\lambda m\\ -\lambda m& 1 \end{bmatrix}, \qquad  
V=\begin{bmatrix} 4\lambda^{-2} + Q & -2\lambda^{-1} (u-u_x)-\lambda m Q\\
2\lambda^{-1}(u+u_x)+\lambda m Q & -Q \end{bmatrix}, \quad Q=u^2-u_x^2, \quad \lambda \in \C. 
\end{equation*} 
Performing the gauge transformation $\Phi=\textrm{diag}(\frac{e^{\frac x2}}{\lambda}, e^{-\frac x2}) \Psi$ results in a simpler $x$-equation
\begin{equation}\label{eq:xLax}
\Phi_x=\begin{bmatrix}0 & h\\
-z g& 0 \end{bmatrix} \Phi, \qquad    g=\sum_{j=1} ^ng_j \delta_{x_j}, \qquad h=\sum_{j=1} ^nh_j \delta_{x_j}, 
\end{equation} 
where $g_j=m_j e^{-x_j}, \, h_j =m_j e^{x_j}, \, z=\lambda^2 $. For future use note that $g_jh_j=m_j^2$.

We will be interested in solving \eqref{eq:xLax} 
subject to boundary conditions $\Phi_1(-\infty)=0, \, \Phi_2(+\infty)=0$.  
To make the boundary value problem 
\begin{equation}\label{eq:xLaxBVP}
\Phi_x=\begin{bmatrix}0 &h\\
-z g& 0 \end{bmatrix} \Phi, \qquad \Phi_1(-\infty)=\Phi_2(+\infty)=0, 
\end{equation}
well posed we need to define the multiplication of the measures $h$ and $g$ by $\Phi$.  
Guided by the results of Appendix \ref{lax_mch} 
we require that $\Phi$ be left continuous and we define the terms $\Phi_{a}\delta_{x_j}=\Phi_a(x_j)\delta_{x_j}, a=1,2$.  This choice makes the Lax pair well defined as a distributional Lax pair and, as it is shown in the Appendix \ref{lax_mch},  
the compatibility condition of the $x$ and $t$ components of the Lax pair indeed implies \eqref{eq:xmODE}.

The solution $\Phi$ is a piecewise constant function which, for convenience, we can normalize by setting $\Phi_2(-\infty)=1$.  The distributional boundary value problem \eqref{eq:xLaxBVP} is in our special case of the discrete measure $m$ equivalent to a finite difference equation. 
\begin{lemma}\label{lem:forwardR}
Let $q_k=\Phi_1(x_{k}+), \, ~p_k=~\Phi_2(x_{k}+),$ then the difference form of the 
boundary value problem reads: 
\begin{equation} \label{dstring}
\begin{gathered}
\begin{aligned}
     q_{k}-q_{k-1}&=h_kp_{k-1}, & 1\leq k\leq n, \\
     p_{k}-p_{k-1}&=-z g_kq_{k-1},& 1\leq k\leq n,\\
      q_0=0, \quad  p_0=1&, \quad p_{n}=0.  &  
  \end{aligned}
\end{gathered}
\end{equation} 
\end{lemma} 
An easy proof by induction leads to the following corollary.  
\begin{corollary} \label{cor:pq-degrees} 

$q_{k}(z)$  is a polynomial of degree $\lfloor\frac{k-1}{2}\rfloor$ in $z$, and $p_{k}(z)$ is a polynomial of degree $\lfloor\frac{k}{2}\rfloor$, respectively. 

\end{corollary} 

\begin{remark} Note that the difference form of the boundary value problem admits a simple matrix presentation, namely a $2\times 2$ matrix encoding of \eqref{dstring}
\begin{equation}\label{transition}
\begin{bmatrix}
  q_{k}\\
  p_{k}
\end{bmatrix}
=T_k\begin{bmatrix}
  q_{k-1}\\
  p_{k-1}
\end{bmatrix}, \qquad\qquad
T_k=\begin{bmatrix}
  1& h_k\\
  -z g_k&1
\end{bmatrix}.
\end{equation}
We point out that the transition matrix $T_k$ is different from the difference equation for the inhomogeneous 
string boundary value problem  $D^2 v =- zg v, \, v(0)=v(1)=0$ discussed in \cite{ls-cubicstring} (Appendix A) for which 
$T_k=\begin{bmatrix}
  1& l_{k-1}\\
  -z g_k&1-zg_kl_{k-1} 
\end{bmatrix}, $ 
thus an element of the group $SL_2(\C)$.  
\end{remark}

We can obtain more precise information about polynomials $p_k, q_k$ by studying directly the solutions to the initial value problem 
\begin{equation}\label{eq:xLaxIVP}
\Phi_x=\begin{bmatrix}0 &h\\
-z g& 0 \end{bmatrix} \Phi, \qquad \Phi_1(-\infty)=0, \quad \Phi_2(-\infty)=1,  
\end{equation}
with the same rule regarding the multiplication of discrete measures $g,h$
by piecewise smooth, left-continuous, functions $f$ as specified above.  
With this proviso expressions like $$ \int_{-\infty}^x f(\xi)\, g(\xi)  d\xi\stackrel{def}{=}
\int\limits_{\xi <x}f(\xi) \, g(\xi) d\xi $$ uniquely define 
piecewise constant functions which we choose to be left continuous.  The same applies to 
iterated integrals over the regions $\{\xi_1<\xi_2<\dots \xi_k <x\}$.   
For example 
\begin{equation*} 
    \int \limits_{\xi_1<\xi_2<x}  f(\xi_1) h(\xi_1)d\xi_1\, g(\xi_2)d\xi_2
\end{equation*} 
 is well defined.  With this notation in place we obtain the following 
 characterization of $\Phi_1(x)$ and $\Phi_2(x)$.  
 \begin{lemma} \label{lem:seriessol-}  Let us set 
 \begin{equation*} 
             \Phi_1(x)=\sum\limits_{0\leq k} \Phi_1^{(k)}(x) z^k, \quad \qquad \Phi_2(x)=            \sum\limits_{0\leq k} \Phi_2^{(k)}(x) z^k.  
 \end{equation*}
Then 
\begin{equation*} 
    \Phi_1^{(0)}(x)=\int\limits_{\eta_0 <x}h(\eta_0) d\eta_0, \quad \qquad \Phi_2^{(0)}(x)=1
\end{equation*} 
~for $k=0$, otherwise
\begin{subequations} 
\begin{align}
\Phi_1^{(k)}(x)&=(-1)^k \int\limits_{\eta_0<\xi_1<\eta_1<\dots<\xi_k<\eta_k<x}
\big[\prod_{p=1}^k h(\eta_p)g(\xi_p)\big]   h(\eta_0)\, \,  d\eta_0 d\xi_1\dots d\eta_k,  \,  &\\
\Phi_2^{(k)}(x)&=(-1)^k \int\limits_{\xi_1<\eta_1<\dots<\xi_k<\eta_k<x}
\big[\prod_{p=1}^k g(\eta_p)h(\xi_p)\big] \,\,  d\xi_1\dots d\eta_k.   \, &
\end{align} 
\end{subequations}
If the points of the support of  the discrete measure $g$ (and $h$) are ordered $x_1<x_2<\dots<x_n$ 
then 
 \begin{subequations} 
\begin{align}
\Phi_1^{(k)}(x)&=(-1)^k \sum_{\substack{j_0<i_1<j_1<\dots< i_k< j_k\\x_{j_k}<x}}
\, \big[\prod_{p=1}^k h_{j_p}g_{i_p}\big] \,\,  h_{j_0} ,  \,  &\\
\Phi_2^{(k)}(x)&=(-1)^k\sum_{\substack{i_1< j_1<\dots< i_k< j_k\\x_{j_k}<x}}
\, \big[\prod_{p=1}^k g_{j_p}h_{i_p}\big] \,\,   .   \, &
\end{align} 
\end{subequations} 
\end{lemma}
\begin{proof}
First we observe that solving equation \eqref{eq:xLaxIVP} is 
equivalent to solving the system of integral equations: 
\begin{equation*} 
    \Phi_1(x)=\int\limits_{\xi<x} \Phi_2(\xi)h(\xi) d\xi, \qquad 
    \Phi_2(x)=1-z \int\limits_{\xi<x} \Phi_1(\xi)g(\xi) d\xi, 
\end{equation*} 
with in turn implies
\begin{align*} 
    \Phi_1(x)&=\int\limits_{\eta_0 <x}h(\eta_0) d\eta_0-z \int\limits_{\xi_1<\eta_1<x} h(\eta_1) g(\xi_1)\Phi_1(\xi_1)  d\xi_1 d\eta_1, \\
    \Phi_2(x)&=1-z \int\limits_{\xi_1<\eta_1<x} g(\eta_1) h(\xi_1) \Phi_2(\xi_1)d\xi_1 d\eta_1.  
\end{align*} 
Elementary iterations yield the final result in the integral form.  Finally, 
once the ordering conditions is in place, the evaluation of integrals as sums 
follows.  
\end{proof} 
We will now introduce a multi-index notation, later used to facilitate writing 
explicit formulas for peakon solutions, but also helpful in capturing the 
properties of solutions to \eqref{dstring}.  A similar notation 
turned out to be very helpful in stating and proving the Canada Day Theorem 
in \cite{gls} (see also \cite{hls}).  

The formulas in Lemma \ref{lem:seriessol-} involve a choice of $j$-element index sets $I$ and~$J$
from the set $[k] = \{ 1,2,\dots,k \}$.
We will use the notation
$\binom{[k]}{j}$ for the set of all $j$-element subsets of $[k]$, listed in increasing order; for example $I\in \binom{[k]}{j}$ means that 
$I=\{i_1, i_2,\dots, i_j\}$ for some increasing sequence $i_1 < i_2 < \dots < i_j\leq ~k$. 
Furthermore, given the multi-index $I$ we will abbreviate $g_I=g_{i_1}g_{i_2}\dots g_{i_j}$ etc.
 
\begin{definition}\label{def:bigIndi} Let $I,J \in \binom{[k]}{j}$, or $I\in \binom{[k]}{j+1},J \in \binom{[k]}{j}$.  
\mbox{}

Then  $I, J$ are said to be \emph{interlacing} if 
\begin{equation*}
  \label{eq:interlacing}
    i_{1} <j_{1} < i_{2} < j_{2} < \dotsb < i_{j} <j_{j}
\end{equation*}
or, 
\begin{equation*}
    i_{1} <j_{1} < i_{2} < j_{2} < \dotsb < i_{j} <j_{j}<i_{j+1}, 
\end{equation*}
in the latter case.  
We abbreviate this condition as $I < J$ in either case, and, furthermore, 
use the same notation, that is $I<J$,   for $I\in \binom{[k]}{1}, J \in \binom{[k]}{0}$.
\end{definition} 
\begin{remark} 
We point out that the multi-indices $I, J$ satisfying $I<J$ 
are called in \cite{gls} \emph{strictly interlacing}.   
\end{remark} 
\begin{corollary}\label{cor:solqkpk}
Let $q_k=\Phi_1(x_{k}+), \, ~p_k=~\Phi_2(x_{k}+),$ then the difference form of the 
initial value problem \eqref{eq:xLaxIVP} reads: 
\begin{equation} \label{dstringIVP}
\begin{gathered}
\begin{aligned}
     q_{k}-q_{k-1}&=h_kp_{k-1}, & 1\leq k\leq n, \\
     p_{k}-p_{k-1}&=-z g_kq_{k-1},& 1\leq k\leq n,\\
      q_0=0, &\quad  p_0=1 & 
  \end{aligned}
\end{gathered}
\end{equation} 
whose unique solution (restoring the dependence on $z$ in $q_k$ and $p_k$) is given by: 
\begin{subequations}
\begin{align}
q_k(z)&=
\sum_{j=0}^{\lfloor\frac{k-1}{2}\rfloor}\Big(\sum_{\substack{I\in \binom{[k]}{j+1}, J\in \binom{[k]}{j}\\ I<J}}
\, h_Ig_J\Big) (-z)^j,  \\
p_{k}(z)&=1+\sum_{j=1}^{\lfloor\frac{k}{2}\rfloor}\Big(\sum_{\substack{I,J \in \binom{[k]}{j}\\ I<J}} h_I g_J
\, \Big)(-z)^j.  
\end{align}
\end{subequations} 
\end{corollary}

Our next goal is to study the spectrum of the boundary value problem 
\eqref{eq:xLaxBVP}. 
\begin{definition} 
A complex number z is an \textit{eigenvalue}  of the boundary value problem 
\eqref{eq:xLaxBVP} if there exists a solution $\{q_k, p_k\}$ to \eqref{dstringIVP} for 
which $p_{n}(z)=0$.   The set of all eigenvalues is the \textit{spectrum}  of the boundary value problem \eqref{eq:xLaxBVP}.  
\end{definition} 

The relevance of the spectrum of \eqref{eq:xLaxBVP} is captured in the following lemma 
which follows from examining the $t$ part of the Lax pair \eqref{eq:xtLax} in the region $x>x_{n}$.  
\begin{lemma}\label{lem:t-evolution of qp} 
Let $\{q_k, p_k\}$ satisfy the system of difference equations \eqref{dstringIVP}.  Then 
\begin{equation}\label{eq:tderqp}
  \dot q_{n}=\frac{2}{z}q_{n}-\frac{2L}{z}\,p_{n}, \qquad \dot p_{n}=0,
\end{equation}
where $L=\sum_{j=1}^{n}h_j$. Thus $p_{n}(z)$ is independent of time and, in particular, its zeros, i.e. the spectrum, are time invariant.
\end{lemma} 
Since Corollary \ref{cor:solqkpk} gives an explicit form 
of $p_n(z)$ we can easily identify the constants of motion implied by 
isospectrality of the boundary value problem \eqref{eq:xLaxBVP}.  
\begin{lemma} \label{lem:constants} 
The quantities 
\begin{equation*} 
M_j=\sum_{\substack{I,J \in \binom{[n]}{j}\\ I<J}} h_I g_J, \qquad 1\leq j\leq \lfloor\frac{n}{2}\rfloor
\end{equation*}
form a set of $\lfloor\frac{n}{2}\rfloor$  constants 
of motion for the system \eqref{mCH_ode}.  
\end{lemma}
\begin{example}
  Let us consider the case $n=4$.  Then the constants of motion, written in the original variables $(m_j, x_j)$, with positions $x_j$ satisfying $x_1<x_2<x_3<x_4$, are
  \begin{align*}
    M_1&=m_1m_2e^{x_1-x_2}+m_1m_3e^{x_1-x_3}+m_1m_4e^{x_1-x_4}+m_2m_3e^{x_2-x_3}+m_2m_4e^{x_2-x_4}+m_3m_4e^{x_3-x_4},\\
    M_2&=m_1m_2m_3m_4e^{x_1-x_2+x_3-x_4}.
  \end{align*}
\end{example}

\section{Forward map: spectrum and spectral data} \label{sec:FSM}
We will characterize the spectrum of the boundary value problem \eqref{eq:xLaxBVP}, or equivalently, \eqref{dstring}
by ~associating it with 
the \textit{Weyl function} 
\begin{equation}\label{eq:defWeyl}
W(z)=\frac{q_{n}(z)}{p_{n}(z)}.  
\end{equation} 
The remainder of this section is devoted to the proof of the following 
theorem characterizing our boundary value problem in terms of $W(z)$.  
\begin{theorem} \label{thm:W} 
$W(z)$ is a (shifted) Stieltjes transform of a positive, discrete measure $d\mu$ with support 
inside $\R_+$.  More precisely: 
\begin{equation}
W(z)=c+\int \frac{d\mu(x)}{x-z}, \qquad d\mu=\sum_{i=1}^{\lfloor \frac n2 \rfloor} 
b_j \delta_{\zeta _j}, \qquad 0<\zeta_1<\dots< \zeta_{\lfloor \frac n2 \rfloor}, \qquad 0<b_j,   \quad 1\leq j\leq \floor{\frac n2}, 
\end{equation}
~ where 
$c>0$ when $n$ is odd 
and $c=0$ when $n$ even.  

\end{theorem} 
The next corollary describes the properties of the spectrum.  
\begin{corollary} \label{cor:spectrum} 
\mbox{}
\begin{enumerate} 
\item The spectrum of the boundary value problem \eqref{eq:xLax} is positive and simple. 
\item  
$W(z)=c+\sum_{j=1}^{\lfloor \frac n2 \rfloor} \frac{b_j}{\zeta_j-z}$, where 
 all residues satisfy $b_j>0$ and $c\geq0$.  
 \end{enumerate} 
 
 \end{corollary}
 The strategy of the proof of Theorem \ref{thm:W} is to show that 
 $W$ has a continued fraction expansion of Stieltjes's type, the term 
 explained below.  
 We start by reformulating 
 the recurrence relation \eqref{dstringIVP}.  
 \begin{lemma}\label{lem:wdstring}
 Let $\{q_k, p_k\}$ be the solution to \eqref{dstringIVP} and let $w_{2k}=\frac{q_{k}}{p_{k}}, w_{2k-1}=\frac{q_{k-1}}{p_{k}}$.  
 Then 
 \begin{subequations} 
 \begin{align}
    w_1=0, \qquad w_{2k}&=(1+zm_k^2)w_{2k-1}+h_k, \quad &&1\leq k\leq  n \label{eq:recweven}\\
    \qquad \frac{1}{w_{2k}}&=\frac{1}{w_{2k+1}}+zg_{k+1},  \quad &&1\leq k\leq n-1 \label{eq:recwodd}
\end{align}
\end{subequations}
  \end{lemma} 
  \begin{proof}
  The first line follows readily by rewriting the 
  first line of \eqref{dstringIVP} 
  as
  \begin{equation*}
\frac{q_{k}}{p_{k}}-\frac{q_{k-1}}{p_{k}}=h_k\frac{(p_{k-1}-p_{k})}{p_{k}}+h_k,
\end{equation*}
then using the second equation of \eqref{dstringIVP} to eliminate $p_{k-1}-p_{k}$, on the way 
employing the relation $m_k^2=g_k h_k$, and finally rewriting the result using the 
definition of $w_{2k}$ and $w_{2k-1}$.  The condition $w_1=0$ corresponds to the 
boundary condition at index $k=1$, recalling that $w_1=\frac{q_0}{p_1}=0$ because $q_0=0, p_1=1$.  The second line follows from the second formula in 
\eqref{dstringIVP}. 
\end{proof}   
\begin{remark} The recurrence in Lemma \ref{lem:wdstring} can be viewed as 
the recurrence on the Weyl functions corresponding to shorter strings 
obtained by truncating at the index $k$.  Then $W_{2k}$ is precisely the 
Weyl function corresponding to the measures $\sum_{j=1}^k h_j \delta_{x_j}$ and $\sum_{j=1}^k g_j \delta_{x_j}$, 
while $W_{2k-1}$ corresponds to the measures $\sum_{j=1}^{k-1}  h_j \delta_{x_j}$ and $\sum_{j=1}^k g_j \delta_{x_j}$ respectively.  
\end{remark} 
Before we state the second lemma we 
briefly review  some old results 
of T. Stieltjes, appropriately adapted to our setup.  More specifically, the 
following description of rational functions follows from general results proved by T. Stieltjes in his famous memoir \cite{stieltjes}.  

\begin{theorem} [T. Stieltjes] \label{thm:Stieltjes} Any rational function $F(z)$ admitting the 
integral representation 
\begin{equation}\label{eq:Stieltjesintegral}
F(z)=c+\int \frac{d\nu(x)}{x-z}, 
\end{equation} 
where $d\nu(x)$ is the (Stieltjes) measure corresponding to the piecewise constant 
non-decreasing function $\nu(x)$ with finitely many jumps in $\R_+$ 
has a finite (terminating) continued fraction expansion
\begin{equation}\label{eq:Stieltjescf}
F(z)=c+\cfrac{1}{a_1 (-z)+\cfrac{1}{a_2+\cfrac{1}{a_3(-z)+\cfrac{1}{\ddots}}}}, 
\end{equation}
where all $a_j>0$ and, conversely, any rational function with this type of 
continued fraction expansion has the integral representation \eqref{eq:Stieltjesintegral}. 
\end{theorem} 

We will refer to the integral representation \eqref{eq:Stieltjesintegral} as 
the \textit{shifted Stieltjes transform} of a measure $dv(x)$.  Now we are ready to 
state the second lemma.  

\begin{lemma}\label{lem:spectralwks}
Given $h_j>0, h_jg_j=m_j^2>0, 1\leq j\leq n$, let $w_j$s satisfy the recurrence relations of Lemma \ref{lem:wdstring}.  
Then $w_j$s are shifted Stieltjes transforms of finite, discrete Stieltjes measures supported on $\R_+$, with nonnegative shifts.  More precisely: 
\begin{align*} 
w_{2k-1}&=\int \frac{d\mu^{(2k-1)}(x)}{x-z}, \\
w_{2k}&=c_{2k}+\int \frac{d\mu^{(2k)}(x)}{x-z}, 
\end{align*} 
where $c_{2k}>0$ when $k$ is odd, otherwise, $c_{2k}=0$.  Furthermore, 
the number of points in the support $d\mu^{(2k)}(x)$ and $d\mu^{(2k-1)}$ is 
$\lfloor\frac{k}{2}\rfloor$.  
\end{lemma} 
\begin{proof} The proof proceeds by induction on $k$.   
The base case $k=1$ is trivial since $w_1=0$ while $w_{2}=h_1$ by the first 
equation in Lemma \ref{lem:wdstring} confirming that $c_2>0$.  
Suppose now the claim is valid for the index $k$.  Thus 
$w_{2k-1}$ and $w_{2k}$ are shifted Stieltjes transforms of some 
measures $d\mu^{(2k-1)}$ and $d\mu^{(2k)}$, both being finite, discrete and 
supported on $\R_+$.  We now solve \eqref{eq:recwodd} for $w_{2k+1}$ 
obtaining: 
\begin{equation*}
w_{2k+1}=\cfrac{1}{-zg_{k+1}+\frac{1}{w_{2k}}}, 
\end{equation*}
then use the induction hypothesis, which implies that 
$w_{2k}$ has the form \eqref{eq:Stieltjescf}, resulting in the continued fraction 
expansion: 
\begin{equation*}
w_{2k+1}=\cfrac{1}{-zg_{k+1}+\cfrac{1}{c+\cfrac{1}{a_1 (-z)+\cfrac{1}{a_2+\cfrac{1}{a_3(-z)+\cfrac{1}{\ddots}}}}}}\, .  
\end{equation*}
If $c>0$ then $w_{2k+1}$ has already the required form.  If, on the other hand, $c=0$ the first term in the continued fraction expansion is $-(g_{k+1}+a_1)z$ and 
since $a_1>0$ the coefficient is positive.  Thus $w_{2k+1}$ satisfies 
the conditions of Stieltjes's Theorem \ref{thm:Stieltjes} and, as a result, $w_{2k+1}$ is 
the Stieltjes transform with zero shift of a finite, discrete measure, say $d\mu^{2k+1}(x) $ supported 
on $\R_+$.  In either case 
\begin{equation*}
w_{2k+1}=\int d\mu^{2k+1}(x)(-\frac 1z)+\mathcal{O}(\frac{1}{z^2}), \qquad z \to \infty, 
\end{equation*} 
where $\int d\mu^{2k+1}(x)=\frac{1}{g_{k+1}}$ if $c>0$ and $\int d\mu^{2k+1}(x)=\frac{1}{g_{k+1}+a_1}$ if $c=0$.  Let us now examine 
equation \eqref{eq:recweven}, shifting $k\to k+1$.  First, we have 
\begin{equation*}
w_{2k+2}=(1+zm^2_{k+1})w_{2k+1}+h_{k+1}, 
\end{equation*}
and, upon using the integral representation for $w_{2k+1}$ we obtain: 
\begin{align*}
w_{2k+2}=\int \frac{(1+xm_{k+1}^2) d\mu^{(2k+1)}(x)}{x-z} +h_{k+1}-m_{k+1}^2 
\int d\mu^{(2k+1)}(x)\\=
\int \frac{(1+xm_{k+1}^2) d\mu^{(2k+1)}(x)}{x-z} +h_{k+1}\left (1-g_{k+1} 
\int d\mu^{(2k+1)}(x)\right). 
\end{align*} 
If the shift $c$ in the formula for $w_{2k}$ is positive then $\int d\mu^{(2k+1)}(x)=\frac{1}{g_{k+1}}$, as remarked earlier, and the shift in the formula 
for $w_{2k+2}$ is $0$.  When $c=0$, $\int d\mu^{(2k+1)}(x)=\frac{1}{g_{k+1}+a_1} <\frac{1}{g_{k+1}} $ and then the shift is positive since $\left (1-g_{k+1} 
\int d\mu^{(2k+1)}(x)\right)>0$.  This proves the integral representation for 
$w_{2k+2}$ and shows that the shift alternates between $0$ and positive numbers, depending on whether 
$k$ is even or odd as claimed since $c_2>0$.  Finally, the number of 
the points in the support of $d\mu^{(2k)}(x)$ and $d\mu^{(2k-1)}(x)$ follows from Corollary \ref{cor:pq-degrees}.   
\end{proof}

Now, with all the preparation, the proof of Theorem \ref{thm:W} follows 
readily from Lemma \ref{lem:spectralwks} by observing that 
\begin{equation}\label{eq:WrepS}
W(z)=c_{2n}+\int \frac{d\mu^{(2n)}(x)}{x-z}, \qquad d\mu^{(2n)}=\sum_{j=1}^{\lfloor \frac n2 \rfloor} b_j^{(2n)}\delta_{\zeta_j}.   
\end{equation}
This concludes the spectral characterization of the boundary value problem \eqref{eq:xLaxBVP}, or equivalently \eqref{dstring}.  

\section{Inverse problem}\label{sec:IP}
\subsection{A solution by recursion}  
The inverse problem associated with the boundary value problem \eqref{eq:xLaxBVP} can be stated: given positive constants $m_j,\, 1\leq j\leq n$, and a rational function $W(z)$ with integral representation \eqref{eq:WrepS}, we seek to invert the map $S: \{x_1, x_2, \dots x_n\}\longrightarrow W$.  

To solve the 
inverse problem we proceed in two stages: 
first we 
reconstruct the positive coefficients $g_j, h_j$ such that $g_jh_j=m_j^2$ then 
we use the relation $\frac{h_j}{g_j}=e^{2x_j}$ to determine $x_j$.  In this section 
we concentrate on the first stage.

The reconstruction of $h_j,g_j$ amounts to solving recurrence relations \eqref{eq:recweven} and \eqref{eq:recwodd} 
following the steps below: 
\begin{enumerate} 
\item starting with $w_{2n}=W(z)$ define $h_n=w_{2n}(-\frac{1}{m_n^2}), \, \, g_n=\frac{m_n^2}{h_n}$ and solve 
\begin{equation*}
w_{2n}=(1+zm_n^2)w_{2n-1}+h_n, \qquad \frac{1}{w_{2n-2}}=\frac{1}{w_{2n-1}}+zg_n, 
\end{equation*}
for $w_{2n-1}$ and $w_{2n-2}$;
\item restart the procedure from $w_{2n-2}$ shifting $n\rightarrow n-1$. 
\end{enumerate} 
We remark that the procedure encodes solving \eqref{eq:recweven}, \eqref{eq:recwodd}
backwards. 
However, for the procedure to make sense, $w_{2n-2}$ needs to be of the form 
\eqref{eq:WrepS}.   
Let us therefore turn to analyzing $w_{2n-2}$.  
First, from the recurrence relation we easily get 
\begin{equation*} 
h_n=c_{2n}+\int\frac{d\mu^{(2n)}(x)}{x+\frac{1}{m_n^2}}=c_{2n}+m_n^2\int d\mu^{(2n-1)}(x), 
\end{equation*}
where 
$d\mu^{(2n-1)}(x)=\frac{d\mu^{(2n)}(x)}{1+m_n^2 x}$, while solving for $w_{2n-1}$ yields
\begin{equation*}
w_{2n-1}(z)=\int \frac{d\mu^{(2n-1)}(x)}{x-z}.  
\end{equation*} 
Thus by Stieltjes's theorem \ref{thm:Stieltjes} 
\begin{equation*}
w_{2n-1}(z)=\cfrac{1}{a_1 (-z)+\cfrac{1}{a_2+\cfrac{1}{a_3(-z)+\cfrac{1}{\ddots}}}}
\end{equation*}
for some $a_j>0$. Next, we write $$w_{2n-2}=\cfrac{1}{zg_n+\cfrac{1}{w_{2n-1}}}=\cfrac{1}{(g_n-a_1)z+\cfrac{1}{a_2+\cfrac{1}{a_3(-z)+\cfrac{1}{\ddots}}}}$$ and observe that for $w_{2n-2}$ 
to have the spectral representation \eqref{eq:WrepS} $g_n-a_1$ must be negative or 
$0$.  However, 
\begin{equation*}
\frac{1}{a_1}-\frac{1}{g_n}=\int d\mu^{(2n-1)}(x)-\frac{h_n}{m_n^2}=
-\frac{c_{2n}}{m_n^2}, 
\end{equation*} 
hence $g_n-a_1 \leq 0$, which proves the existence of 
the spectral representation \eqref{eq:WrepS} for $w_{2n-2}$ for some 
measure $d\mu_{2n-2}$ supported on a finite number of points in $\R_+$.  

Similar to the content of Lemma \ref{lem:spectralwks} we have the following dichotomy: if $c_{2n}=0$, which by the same Lemma happens if $n$ is even, 
the support of $w_{2n-2}$ has one less point in the spectrum of the corresponding 
measure compensated by the appearance of non-zero $c_{2n-2}$.  If, on the other 
hand, $n$ is odd, in which case $c_{2n}>0$, then $c_{2n-2}=0$ and the number of points in the support of $d\mu^{(2n-2)}$ does not differ from that of $d\mu^{(2n)}$.  
In either case, by iterating, one reaches $w_2$ which is a positive constant equal by definition to $h_1$ and the iteration stops.  We conclude the discussion of the solution to the 
inverse problem of recovering $\{g_j, h_j\}$ by recursion with the following 
theorem. 

\begin{theorem} \label{thm:contInverse}
The inverse spectral problem is uniquely solvable 
for any positive masses $m_j$ and 
the inverse map is continuous both 
with respect to the masses $m_j$ as well as the 
spectral data $\{\zeta_1<\zeta_2<\cdots <\zeta_{\floor{\frac n2}}; b_1, b_2,\cdots, b_{\floor{\frac n2}}; c\}$.  
\end{theorem} 
\begin{proof} 
The uniqueness follows by construction of the inverse map.  As discussed earlier 
there are no obstructions to invertibility present at each stage of the recursion and the updated spectral data is obtained by evaluation and algebraic inversions of 
of continuous functions  (Weyl functions) at points ($-\frac{1}{m_j^2}$) where 
those Weyl functions are strictly positive.  
\end{proof} 

\subsection{A solution by interpolation; basic ideas} 
The iteration proposed above requires $2n-2$ steps to reach $w_2$, 
each step leading to a new input rational function $w_j$.  The formulas 
for $h_j$ get increasingly more complicated and a natural 
question presents itself: can one compute $h_j$ using directly the spectral data 
$c_{2n}$ and $d\mu^{(2n)}$?  The answer is affirmative and this section 
outlines the main steps of the construction leaving 
the detailed formulas for the following sections in which we present 
a complete solution to the peakon problem \eqref{mCH_ode}.  

First we give a brief summary of main ideas behind the solution 
by interpolation.  Let us rewrite \eqref{transition} in terms of the Weyl function $W=w_{2n} $ as 
\begin{equation} \label{eq:Papprox1}
\begin{bmatrix} W(z) \\ 1 \end{bmatrix} =T_n(z)T_{n-1}(z)  \dots T_{n-k+1}(z)\begin{bmatrix} \frac{q_{n-k}(z)}{p_{n}(z)}\\ \frac{p_{n-k}(z)}{p_{n}(z)} \end{bmatrix}. 
\end{equation}
Clearly, the transpose of the matrix of cofactors of each $T_j(z)$ 
is $\begin{bmatrix} 1 &- h_j\\zg_j &1 \end{bmatrix}\stackrel{\textrm {def}}{=}C_{n-j+1}(z)$,  
which allows one to express equation \eqref{eq:Papprox1}  as 
\begin{equation*}
C_{k}(z)\dots C_1(z)\begin{bmatrix} W(z) \\ 1 \end{bmatrix} =\det (T_n(z))\det (T_{n-1}(z))  \dots \det(T_{n-k+1}(z))\begin{bmatrix} \frac{q_{n-k}(z)}{p_{n}(z)}\\ \frac{p_{n-j}(z)}{p_{n}(z)} \end{bmatrix}, 
\end{equation*}
which, recalling that $\det T_j(z)=1+zm_j^2$ and that the roots of 
$p_{n}(z)$ are all positive, implies ~
\begin{equation} \label{eq:Papprox2}
\left( C_{k}(z)\dots C_1(z)\begin{bmatrix} W(z) \\ 1 \end{bmatrix}\right)\Big |_{z=-\frac{1}{m_{n-i+1}^2}}=0,  \qquad \textrm{ for any } 1\leq i\leq k.  
\end{equation}

\begin{theorem} \label{thm:CJ interp}
Let the matrix of products of $C$s in equation \eqref{eq:Papprox2} be denoted by ~$\begin{bmatrix} a_k(z)&b_k(z)\\c_k(z)&d_k(z) \end{bmatrix}\stackrel{\textrm {def}}{=}\hat S_k(z)$.  Then the polynomials $a_k(z), b_k(z),c_k(z),d_k(z)$ 
solve the following interpolation problem: 
\begin{subequations}\label{eq:Papprox}
\begin{align}
&a_k(-\frac{1}{m_{n-i+1}^2})W(-\frac{1}{m_{n-i+1}^2})+b_k(-\frac{1}{m_{n-i+1}^2})
=0, \qquad 1\leq i\leq k, \\
&\deg a_k=\big\lfloor \frac{k}{2} \big \rfloor, \qquad \deg b_k=\big \lfloor \frac{k-1}{2} \big\rfloor, \qquad a_k(0)=1, \\
\notag\\
&c_k(-\frac{1}{m_{n-i+1}^2})W(-\frac{1}{m_{n-i+1}^2})+d_k(-\frac{1}{m_{n-i+1}^2})
=0, \qquad 1\leq i\leq k,  \\
&\deg c_k=\big \lfloor \frac{k+1}{2}\big \rfloor, \qquad \deg d_k=\big \lfloor \frac{k}{2}\big \rfloor, \qquad c_k(0)=0, \quad d_k(0)=1.   
\end{align}
\end{subequations} 
\end{theorem} 
\begin{proof} 
The approximation statements follow directly from \eqref{eq:Papprox2}, while  
the degrees follow, by induction, from the definition of $C_j$ and 
the formula for $\hat S_k$.  
\end{proof}
\begin{remark} The interpolation \eqref{eq:Papprox} is an example of a \textit {Cauchy-Jacobi interpolation problem}\cite{gonvcar1978markov,meinguet1970on,stahl1987existence}, studied as part of a general multi-point Pad\'{e} approximation theory \cite{baker1996pade}.
\end{remark}

Before we solve the interpolation problem it is helpful to understand 
how information about the measures $g$ and $h$ is encoded in the coefficients 
$a_j(z), b_j(z), c_j(z), d_j(z)$.  To this end we define another 
initial value problem, following the general philosophy of scattering theory, this time 
specifying initial conditions at $x=+\infty$.  
\begin{equation}\label{eq:xLaxIVP+}
    \hi_x=\begin{bmatrix}0 &h\\
     -z g& 0 \end{bmatrix} \hi, \qquad \hi_1(+\infty)=1, \quad \hi_2(+\infty)=0,  
\end{equation}
and seeking, in contrast to \eqref{eq:xLaxIVP}, the right-continuous solutions, 
interpreting the products $\hi_{a}\delta_{x_j}$ as ~$\hi_{a}\delta_{x_j}=\hi_a(x_j)\delta_{x_j}, a=1,2$. 
Subsequently we define the (right) boundary value problem: 
\begin{equation}\label{eq:xLaxBVP+}
 \hi_x=\begin{bmatrix}0 &h\\
     -z g& 0 \end{bmatrix} \hi, \qquad \hi_1(-\infty)=0, \quad \hi_2(+\infty)=0, 
\end{equation}
seeking right continuous solutions.  
\begin{remark} 
We refer to \eqref{eq:xLaxBVP+} as the (right) boundary value problem, 
even though it is formally the same boundary value problem as ~$\eqref{eq:xLaxBVP}$
but we stress that the rules of defining the singular operation of multiplication 
of a measure by piecewise-smooth functions has changed, hence, we don't know 
~\textit{a priori} if the boundary value problems are indeed the same.  We will 
establish below that they are. 
\end{remark}

\begin{lemma}\label{lem:forwardR+}
Let $\hat q_j=\hi_1(x_{j'}-), ~\hat p_j=~\hi_2(x_{j'}-)$,  where $j'=n+1-j$. 

\noindent Then the difference form of the 
(right) boundary value problem \eqref{eq:xLaxBVP+} reads: 
\begin{equation} \label{dstring+}
\begin{gathered}
\begin{aligned}
     \hat q_{j}-\hat q_{j-1}&=-h_{j'} \hat p_{j-1}, & 1\leq j\leq n, \\
     \hat p_{j}-\hat p_{j-1}&=z g_{j'}\hat q_{j-1},& 1\leq j\leq n,\\
       \hat p_0&=0, \quad \hat q_{n}=0.   & 
  \end{aligned}
\end{gathered}
\end{equation}
  \end{lemma} 
 The accompanying initial value problem is chosen for the remainder of the 
 discussion to
 have  initial conditions $\hat q_0=1,\, \hat p_0=0$. Furthermore, 
the notation $j'=n+1-j$ (reflection of the interval $[1,n]$, or counting from the right 
end $n$ ) is in force from this point onward.  
\begin{lemma} The difference form of the (right) boundary value problem \eqref{dstring+} can be written in matrix form
\begin{equation}\label{transition+}
\begin{bmatrix}
  \hat q_{j}\\
  \hat p_{j}
\end{bmatrix}
=\widehat T_j\begin{bmatrix}
  \hat q_{j-1}\\
  \hat p_{j-1}
\end{bmatrix}, \qquad\qquad
\widehat T_j=\begin{bmatrix}
  1& -h_{j'}\\
  z g_{j'}&1
\end{bmatrix}, 
\end{equation}
and $C_{j'}(z)$, the transpose of the cofactor matrix of $T_j(z)$ appearing in \eqref{eq:Papprox2}, satisfies 
\begin{equation*}
C_{j}(z)=\hat T_{j}(z), \qquad 1\leq j\leq n,  
\end{equation*}
 and its product $\hat S_k(z)$, also defined in Theorem \ref{thm:CJ interp}, is 
 the transition matrix for the right boundary value problem, namely, 
\begin{equation*} 
\hat S_k(z)=\hat T_k(z) \cdots \hat T_1(z). 
\end{equation*} 

\end{lemma}

\begin{lemma} \label{lem:seriessol+}  
Consider the initial value problem given by equation \eqref{eq:xLaxIVP+} and let us set 
 \begin{equation*} 
             \hi_1(x)=\sum\limits_{0\leq k} \hi_1^{(k)}(x) z^k, \quad \qquad \hi_2(x)=            \sum\limits_{0\leq k} \hi_2^{(k)}(x) z^k.  
 \end{equation*}
Then 
\begin{subequations} 
\begin{align}
\hi_1^{(0)}(x)&=1,    \qquad   \hi_2^{(0)}(x)=0, &\\ 
\hi_1^{(k)}(x)&=(-1)^k \int\limits_{x<\xi_1<\eta_1<\dots<\xi_k<\eta_k}
\big[\prod_{j=1}^k h(\xi_j)g(\eta_j)\big] \,\,  d\xi_1\dots d\eta_k,    \, &1\leq k, \\
\hi_2^{(k)}(x)&=(-1)^{k -1}\int\limits_{x<\eta_0<\xi_1<\eta_1<\dots<\xi_{k-1}<\eta_{k-1}}
g(\eta_0)\big[\prod_{j=1}^{k-1} h(\xi_j)g(\eta_j)\big]   \, \,  d\eta_0 d\xi_1\dots d\eta_{k-1},  \,   &1\leq k, 
\end{align} 
\end{subequations}
where, for $k=1$, $\prod_{j=1}^{0}$ is defined to be $1$ and the 
integration is carried out with respect to $\eta_0$ only.  

Furthermore, 
if the points of the support of $g$ (and $h$) are ordered $x_1<x_2<\dots<x_n$ 
then 
 \begin{subequations} \label{eq:iterativesol}
\begin{align}
\hi_1^{(k)}(x)&=(-1)^k\sum_{\substack{i_1<j_1<\dots<i_k<j_k\\x<x_{i_1}}}
\, \big[\prod_{l=1}^k h_{i_l}g_{j_l}\big] \,\,   ,    \, &\\
\hi_2^{(k)}(x)&=(-1)^{k-1} \sum_{\substack{j_0<i_1<j_1<\dots<i_k<j_k\\x<x_{j_0}}}
\, g_{j_0} \big[\prod_{l=1}^{k-1} g_{i_l}h_{j_l}\big] \,\,  . \,  &
\end{align} 
\end{subequations} 
\end{lemma}
Clearly, by setting $\hat q_k=\hi_1(x_{k'}-), \hat p_k=\hi_2(x_{k'}-)$, 
with the help of \eqref{eq:iterativesol}, 
we obtain the solution to difference equations \eqref{dstring+} to 
initial conditions $\hat q_0=1, \hat p_0=0$.  The procedure can be 
repeated for the case of initial conditions $\hat Q_0=0, \hat P_0=1$, yielding 
a complementary solution to \eqref{dstring+}.  We will skip the intermediate steps
since they are very similar to the computations leading up to Lemma \ref{lem:seriessol+}.  
To state the final result we remark that the map $i\to i'=n+1-i$ 
is a bijection between $[1,k]$ and $[n+1-k, n]$.  This map 
can be lifted to multi-indices $I \in \binom{[1,k]}{j}$ introduced earlier, in particular 
given $I={i_1<i_2<\cdots<i_j} \in \binom{[1,k]}{j}$ let us denote by $I'$ 
its image $\{i_1'>i_2'>\cdots>i_j'\} \in \binom{[n-k+1,n]}{j}$.  

\begin{theorem} \label{thm:explicit S-hat}
Consider the right boundary value problem \eqref{dstring+} 
with its transition matrix
\begin{equation*}
\hat S_k=\hat T_k\cdots \hat T_1.  
\end{equation*}
Then 
\begin{equation*}
\hat S_k=\begin{bmatrix} \hat q_k& \hat Q_k\\ \hat p_k& \hat P_k \end{bmatrix}, \quad 1\leq k\leq n,
\end{equation*}
where 
\begin{subequations}
\begin{align}
\hat q_{k}(z)&=1+\sum_{j=1}^{\lfloor\frac{k}{2}\rfloor}\Big(\sum_{\substack{I,J \in \binom{[k]}{j}\\ I<J}} g_{I'} h_{J'}
\, \Big)(-z)^j, 
&&\hat Q_{k}(z)=-\sum_{j=0}^{\lfloor\frac{k-1}{2}\rfloor}\Big(\sum_{\substack{I\in \binom{[k]}{j+1}, J\in \binom{[k]}{j}\\ I<J}} h_{I'}g_{J'}
\, \Big)(-z)^j, \\
\hat p_k(z)&=
-\sum_{j=1}^{\lfloor\frac{k+1}{2}\rfloor}\Big(\sum_{\substack{I\in \binom{[k]}{j}, J\in \binom{[k]}{j-1}\\ I<J}}
\, g_{I'} h_{J'} \Big) (-z)^j, 
&&\hat P_k(z)=1+
\sum_{j=1}^{\lfloor\frac{k}{2}\rfloor}\Big(\sum_{\substack{I, J\in \binom{[k]}{j}\\ I<J}}
h_{I'} g_{J'} \Big) (-z)^j.  
\end{align}
\end{subequations} 

\end{theorem}

\begin{corollary}\label{cor:gformulas}
Let $\hat S_k$ be the transition matrix for the right 
boundary value problem as specified above. 
 
\hbox{}
\begin{enumerate} 
\item The entries of $\hat S_k$ solve 
the interpolation problems \eqref{eq:Papprox}, that is: $\hat q_k, \hat p_k, \hat Q_k, \hat P_k$ satisfy 
\begin{subequations}\label{eq:Papproxbis}
\begin{align}
&\hat q_k(-\frac{1}{m_{i'}^2})W(-\frac{1}{m_{i'}^2})+\hat Q_k(-\frac{1}{m_{i'}^2})
=0, \qquad 1\leq i\leq k, \label{eq:Papproxbis1}\\
&\deg \hat q_k=\big\lfloor \frac{k}{2} \big \rfloor, \qquad \deg \hat Q_k=\big \lfloor \frac{k-1}{2} \big\rfloor, \qquad \hat q_k(0)=1, \label{eq:Papproxbis1nc}\\
\notag\\
&\hat p_k(-\frac{1}{m_{i'}^2})W(-\frac{1}{m_{i'}^2})+\hat P_k(-\frac{1}{m_{i'}^2})
=0, \qquad 1\leq i\leq k,  \label{eq:Papproxbis2}\\
&\deg \hat p_k=\big \lfloor \frac{k+1}{2}\big \rfloor, \qquad \deg \hat P_k=\big \lfloor \frac{k}{2}\big \rfloor, \qquad \hat p_k(0)=0, \quad \hat P_k(0)=1.   \label{eq:Papproxbis2nc}
\end{align}
\end{subequations} 
 \item Given $f(z)\in \C[z]$, let $\hoc{f}$ denote the coefficient of the term of the highest degree . 
Then
\begin{subequations}
\begin{align}
g_{k'}&=\frac{\hoc{\hat p_k}}{\hoc{\hat q_{k-1}}},   &\text{if k  is odd,} \label{eq:inversegodd}\\
g_{k'}&=\frac{\hoc{\hat P_k}}{\hoc{\hat Q_{k-1}}}, &\text{if k is even.}\label{eq:inversegeven}
\end{align} 
\end{subequations} 
\item The right boundary value problem \eqref{eq:xLaxBVP+} has the same spectrum as the left boundary value problem \eqref{eq:xLaxBVP}.  
\end{enumerate} 
\end{corollary}
\begin{proof}

The interpolation problem was stated in Theorem \ref{thm:CJ interp} 
for the matrix elements of $C_{k}C_{(k-1)}\cdots C_{1}$ (and that's where $\hat S_k$ was introduced). 
However, by Theorem \ref{thm:explicit S-hat}, $\hat S_k$ is the same as $\begin{bmatrix}\hat q_k & \hat Q_k\\ \hat p_k& \hat P_k \end{bmatrix}$, hence the first claim. 

To prove the second claim with consider first the case of odd $k$.  
Then, by the formulas in Theorem \ref{thm:explicit S-hat}
we get: 
\begin{align*} 
\hoc{\hat q_{k-1}}&=(-1)^{\frac{k-1}{2}}\sum_{\substack{I,J \in \binom{[k-1]}{\frac{k-1}{2}}\\ I<J}} g_{I'} h_{J'}=(-1)^{\frac{k-1}{2}}g_{1'}h_{2'}g_{3'}\cdots g_{(k-2)'}h_{(k-1)'}, 
\\
\hoc{\hat p_{k}} &=- (-1)^{\frac{k+1}{2}}\sum_{\substack{I \in \binom{[k]}{\frac{k+1}{2}}, J \in \binom{[k]}{\frac{k-1}{2}}\\ I<J}} g_{I'} h_{J'}=(-1)^{\frac{k-1}{2}}g_{1'}h_{2'}g_{3'}\cdots g_{(k-2)'}h_{(k-1)'}g_{k'}, 
\end{align*}
whose ratio gives the desired formula for $g_{k'}$, recalling that $\hat q_0=1$  to cover the case of $k=1$.  The argument 
for even $k$ is similar except that one uses the formulas for the second column 
of $\hat S_k$.  

Finally, to prove the last claim, we observe that the map $i\rightarrow n+1-i$ is a bijection of the set $[n]$.  
Upon comparing Corollary \ref{cor:solqkpk} with the formula for $\hat q_n$ given above we see that $\hat q_{n}(z)=p_{n}(z)$, hence the two boundary value problems 
are equivalent. 
\end{proof}


\subsection{Solving the inverse problem by interpolation}

The inverse problem we are interested 
in solving explicitly can be stated as follows: 
\begin{definition} \label{def:ISP}
Given a rational function (see Theorem \ref{thm:W})
\begin{equation} \label{eq:Wc}
W(z)=c+\int \frac{d\mu(x)}{x-z}, \qquad d\mu=\sum_{i=1}^{\lfloor \frac n2 \rfloor} 
b_j \delta_{\zeta _j}, \qquad 0<\zeta_1<\dots< \zeta_{\lfloor \frac n2 \rfloor}, \qquad 0<b_j,   \quad 1\leq j\leq \floor{\frac n2}, 
\end{equation}
~ where 
$c>0$ when $n$ is odd 
and $c=0$ when $n$ even, as well as 
positive, distinct, constants $m_1,m_2,\dots, m_n$,  
find positive constants $g_j, h_j$, $1\leq j\leq n$, such 
that $g_jh_j=m_j^2$ and the unique solution of the initial value problem: 
\begin{equation*} 
\begin{gathered}
\begin{aligned}
     q_{k}-q_{k-1}&=h_kp_{k-1}, & 1\leq k\leq n, \\
     p_{k}-p_{k-1}&=-z g_kq_{k-1},& 1\leq k\leq n,\\
      q_0=0, &\quad  p_0=1, &  
  \end{aligned}
\end{gathered}
\end{equation*} 
satisfies 
\begin{equation*}
W(z)=\frac{q_n(z)}{p_n(z)}.  
\end{equation*} 
\end{definition}
\begin{remark} The restriction that the constants $m_j$ be distinct has been made to 
facilitate the argument and 
will be eventually relaxed by taking appropriate limits of the generic case (see 
Theorem \ref{thm:minversex}).  
\end{remark}  
The key observation leading to the solution of this inverse problem 
is the realization that the interpolation problem \eqref{eq:Papproxbis}(the same as 
\eqref{eq:Papprox}) has 
a unique solution.  
\begin{theorem}\label{thm:solPapproxbis}
Given a rational function $W(z)$ as above, and positive, distinct 
constants $m_1,m_2,\dots, m_n$, there exist unique 
solutions $\hat q_k, \hat p_k, \hat Q_k, \hat P_k, 1\leq k\leq n$ to the 
interpolations problems \eqref{eq:Papproxbis}.  

Let $z_{i}=-~\frac{1}{m_{i'}^2}, \, 1\leq i\leq k$, then the solution to the first 
interpolation problem \eqref{eq:Papproxbis1}, \eqref{eq:Papproxbis1nc}is 
\begin{equation}\label{eq:hatqhatQsol}
\begin{split}
&\hat q_k(z)+z^{{\lfloor \frac k2 \rfloor}+1}\hat Q_k(z)\\=\frac{1}{D_k}&\det
\begin{bmatrix} 1&z&\dots&z^{\lfloor \frac k2 \rfloor}&z^{{\lfloor \frac k2 \rfloor}+1}&z^{{\lfloor \frac k2 \rfloor}+2}&\cdots&z^k\\
W(z_1)&z_1W(z_1)&\dots&z_1^{\lfloor \frac k2 \rfloor}W(z_1)&1&z_1&\cdots&z_1^{\lfloor \frac{k-1}{2} \rfloor}\\
\vdots&\vdots&\ddots&\vdots&\vdots&\vdots&\ddots&\vdots\\
W(z_k)&z_k W(z_k)&\dots&z_k^{\lfloor \frac k2 \rfloor}W(z_k)&1&z_k&\cdots&z_k^{\lfloor \frac{k-1}{2} \rfloor} 
\end{bmatrix}, 
\end{split}
\end{equation} 
where 
\begin{equation}
D_k=\det \begin{bmatrix} 
z_1W(z_1)&\dots&z_1^{\lfloor \frac k2 \rfloor}W(z_1)&1&z_1&\cdots&z_1^{\lfloor \frac{k-1}{2} \rfloor}\\
\vdots&\ddots&\vdots&\vdots&\vdots&\ddots&\vdots&\\
z_k W(z_k)&\dots&z_k^{\lfloor \frac k2 \rfloor}W(z_k)&1&z_k&\cdots&z_k^{\lfloor \frac{k-1}{2} \rfloor} 
\end{bmatrix}.  
\end{equation}
Likewise, the solution to the second interpolation problem \eqref{eq:Papproxbis2}, \eqref{eq:Papproxbis2nc} is 
\begin{equation}\label{eq:hatphatPsol}
\begin{split}
&\hat P_k(z)+z^{{\lfloor \frac k2 \rfloor}}\hat p_k(z)\\=\frac{1}{E_k}&\det
\begin{bmatrix} 1&z&\dots&z^{\lfloor \frac k2 \rfloor}&z^{{\lfloor \frac k2 \rfloor}+1}&z^{{\lfloor \frac k2 \rfloor}+2}&\cdots&z^k\\
1&z_1&\dots&z_1^{\lfloor \frac k2 \rfloor}&z_1W(z_1)&z_1^2W(z_1)&\cdots&z_1^{\lfloor \frac{k+1}{2} \rfloor}W(z_1)\\
\vdots&\vdots&\ddots&\vdots&\vdots&\vdots&\ddots&\vdots\\
1&z_k &\dots&z_k^{\lfloor \frac k2 \rfloor}&z_kW(z_k)&z_k^2W(z_k)&\cdots&z_k^{\lfloor \frac{k+1}{2} \rfloor}W(z_k)
\end{bmatrix}, 
\end{split}
\end{equation} 
where 
\begin{equation}
E_k=\det \begin{bmatrix} 
z_1&\dots&z_1^{\lfloor \frac k2 \rfloor}&z_1W(z_1)&z_1^2W(z_1)&\cdots&z_1^{\lfloor \frac{k+1}{2} \rfloor}W(z_1)\\
\vdots&\ddots&\vdots&\vdots&\vdots&\ddots&\vdots\\
z_k &\dots&z_k^{\lfloor \frac k2 \rfloor}&z_kW(z_k)&z_k^2W(z_k)&\cdots&z_k^{\lfloor \frac{k+1}{2} \rfloor}W(z_k)
\end{bmatrix}.  
\end{equation}

\end{theorem} 
\begin{proof} Set 
\begin{equation*}
\hat q_k(z)=1+\sum_{j=1}^{\lfloor \frac k2 \rfloor}a_j z^j, 
\qquad \hat Q_k(z)=\sum_{j=0}^{\lfloor \frac {k-1}{2} \rfloor}A_j z^j. 
\end{equation*}
Then the first interpolation problem reads: 
\begin{equation*} 
\sum_{j=1}^{\lfloor \frac k2 \rfloor} (z_i)^j W(z_j) a_j +
\sum_{j=0}^{\lfloor{\frac{k-1}{2}\rfloor}} (z_i)^j A_j=-W(z_i), \qquad 1\leq i \leq k, 
\end{equation*}
whose solution, by virtue of Cramer's rule, can be written in the form 
of equation \eqref{eq:hatqhatQsol}, provided that $D_k\neq 0$.  Likewise, 
the solution to the second interpolation problem can be easily 
deduced by writing 
\begin{equation*}
\hat P_k(z)=1+\sum_{j=1}^{\lfloor \frac k2 \rfloor}B_j z^j, 
\qquad \hat p_k(z)=\sum_{j=1}^{\lfloor \frac {k+1}{2} \rfloor}b_j z^j,  
\end{equation*}
substituting into the interpolation problem \eqref{eq:Papproxbis2} and, again, 
using Cramer's rule, with the same proviso that $E_k\neq 0$.  Thus it remains to 
prove that $D_k$ and $E_k$ are not $0$ under our assumption of 
distinct masses $m_j$.  To this end we derive below explicit formulas for the 
determinants $D_k$ and $E_k$ from which we conclude that none of the 
determinants can be $0$ in view of the non-degeneracy assumption on the masses $m_j$ (see Corollary \ref{cor:DkEknonzero}).  
\end{proof} 
\subsection{Evaluation of determinants}
In this subsection, we will derive explicit formulas for determinants 
appearing in the solution to the interpolation problems (see Theorem \ref{thm:solPapproxbis}).  We begin by introducing some additional notation to 
facilitate the presentation of formulas, reminding the reader that the 
multi-index notation was introduced earlier in the part leading up to 
the definition \ref{def:bigIndi}.  The following notation is in place: 
 we denote $[i,j]=\{i,i+1,\cdots,j\}, \,\,  \binom{[1,K]}{k}=\{J=\{j_1,j_2,\cdots,j_k\}|j_1<\cdots <j_k, j_i\in [1,K]\}$.  Then for two ordered multi-index sets $I, J$  we define 
\begin{align*}
  \mathbf{x}_J&=\prod_{j\in J}x_j, &\Delta_J(\mathbf{x})&=\prod_{i<j\in J}(x_j-x_i), \\
  \Delta_{I,J}(\mathbf{x};\mathbf{y})&=\prod_{i\in I}\prod_{j\in J}(x_i-y_j), 
  &\Gamma_{I,J}(\mathbf{x};\mathbf{y})&=\prod_{i\in I}\prod_{j\in J}(x_i+y_j),  
  \end{align*}
  along with the convention
\begin{align*}
&\Delta_\emptyset(\mathbf{x})=\Delta_{\{i\}}(\mathbf{x})=\Delta_{\emptyset,J}(\mathbf{x};\mathbf{y})=\Delta_{I,\emptyset}(\mathbf{x};\mathbf{y})=\Gamma_{\emptyset,J}(\mathbf{x};\mathbf{y})=\Gamma_{I,\emptyset}(\mathbf{x};\mathbf{y})=1,&\\
&\binom{[1,K]}{0}=1;\qquad\qquad \binom{[1,K]}{k}=0,\ \  k>K.&
\end{align*}
\begin{definition} 
Given two vectors $\mathbf{e}\in \R^k, \mathbf{d}\in \R^l, \, 0\leq l\leq k$ such that 
$e_i+d_j\neq0$ for any pair of indices, a \textit{Cauchy-Vandermonde matrix}\cite{finck1993inversion,gasca1989computation,martinez1998factorizations,martinez1998fast} is a matrix of the form
\begin{equation}\label{eq:CV} 
  CV_k^{(l)}(\mathbf{e},\mathbf{d})=\left(\begin{array}{cccccccc}
    \frac{1}{e_1+d_1}&\frac{1}{e_1+d_2}&\cdots&\frac{1}{e_1+d_l}&1&e_1&\cdots&e_1^{k-l-1}\\
    \frac{1}{e_2+d_1}&\frac{1}{e_2+d_2}&\cdots&\frac{1}{e_2+d_l}&1&e_2&\cdots&e_2^{k-l-1}\\
    \vdots&\vdots&\ddots&\vdots&\vdots&\vdots&\ddots&\vdots\\
    \frac{1}{e_k+d_1}&\frac{1}{e_k+d_2}&\cdots&\frac{1}{e_k+d_l}&1&e_k&\cdots&e_k^{k-l-1}
  \end{array}\right).  
\end{equation}

\end{definition} 
Two special cases are: for $l=0$ the matrix defined by \eqref{eq:CV} is a classical Vandermonde matrix and for $l=k$ it is a classical Cauchy matrix, and for both 
these special cases there exist classical formulas expressing their 
determinants.  Luckily, 
there exists also a compact formula for the determinant of the (generic) Cauchy-Vandermonde matrix \cite{gasca1989computation,martinez1998factorizations,martinez1998fast}:
\begin{equation}\label{eq:detCV}
  \det(CV_k^{(l)}(\mathbf{e},\mathbf{d}))=\frac{\Delta_{[1,k]}(\mathbf{e})\Delta_{[1,l]}(\mathbf{d})}{\Gamma_{[1,k],[1,l]}(\mathbf{e};\mathbf{d})}.
\end{equation}

As a side note we would like to mention that the Cauchy-Vandermonde matrix \eqref{eq:CV} appears naturally as the coefficient matrix of a rational interpolation problem of Lagrange type: given $k$ pairs of interpolation data $(e_1, t_1),\cdots,(e_k,t_k)$, where $e_1,\cdots,e_k$ are different real numbers, find a function
\[
f(x)=\sum_{j=1}^ls_j\frac{1}{x+d_j}+\sum_{j=l+1}^ks_jx^{j-l-1}, 
\]
with $s_i$ to be determined, such that $f(e_i)=t_i, i=1,\cdots,k$.

The interpolation problems \eqref{eq:Papprox} can be viewed as
 slight variations on the theme of rational interpolations problem of Largrange type 
 and to effect the explicit solution of these problems  one is led to a generalization of the  Cauchy-Vandermonde matrix. 
 
 ~\begin{definition} \label{def:mCV}
 Given three  vectors $\mathbf{e}\in \R^k, \mathbf{d}, \mathbf{a} \in \R^l, \,  \, 0\leq l\leq k$ such that 
$e_i+d_j\neq0$ for any pair of indices, a \textit{ modified Cauchy-Vandermonde matrix} is that of the form 
\begin{equation}
  CV_k^{(l,p)}(\mathbf{e},\mathbf{d},\mathbf{a})=\left(\begin{array}{cccccccc}
    \frac{a_1e_1^{p}}{e_1+d_1}&\frac{a_2e_1^{p+1}}{e_1+d_2}&\cdots&\frac{a_le_1^{p+l-1}}{e_1+d_l}&1&e_1&\cdots&e_1^{k-l-1}\\
    \frac{a_1e_2^{p}}{e_2+d_1}&\frac{a_2e_2^{p+1}}{e_2+d_2}&\cdots&\frac{a_le_2^{p+l-1}}{e_2+d_l}&1&e_2&\cdots&e_2^{k-l-1}\\
    \vdots&\vdots&\ddots&\vdots&\vdots&\vdots&\ddots&\vdots\\
    \frac{a_1e_k^{p}}{e_k+d_1}&\frac{a_2e_k^{p+1}}{e_k+d_2}&\cdots&\frac{a_le_k^{p+l-1}}{e_k+d_l}&1&e_k&\cdots&e_k^{k-l-1}
  \end{array}\right),
\end{equation}
with $p\geq0, 0\leq l\leq k, p+l-1\leq k-l$.  
\end{definition} 
 
\begin{theorem}\label{thm:mCV}

  Let $\, 0\, \leq p, \, 0\leq l\leq k $ and $ p+l-1\leq k-l,$ then
\begin{equation}\label{det:gmCV}
  \det(CV_k^{(l,p)}(\mathbf{e},\mathbf{d},\mathbf{a}))=C_{l,p}\frac{\Delta_{[1,k]}(\mathbf{e})\Delta_{[1,l]}(\mathbf{d})}{\Gamma_{[1,k],[1,l]}(\mathbf{e};\mathbf{d})}, 
\end{equation}
where $C_{l,p}=(-1)^{lp+\frac{l(l-1)}{2}}\mathbf{a}_{[1,l]}\cdot \mathbf{d}_{[1,l]}^p\cdot d_1^0 d_2^1\dots d_l^{l-1} $.  
\end{theorem} 
\begin{proof} 
By multilinearity of the determinant we can factor 
all coefficients $a_j$ from the first $l$ rows.  Hence it is sufficient to 
work with $\mathbf{a}=[1, 1, \cdots, 1] $.  Let us drop the reference 
to $\mathbf{a}$ for the remainder of the proof and simply write
\begin{equation}
  CV_k^{(l,p)}(\mathbf{e},\mathbf{d})=\left(\begin{array}{cccccccc}
    \frac{e_1^{p}}{e_1+d_1}&\frac{e_1^{p+1}}{e_1+d_2}&\cdots&\frac{e_1^{p+l-1}}{e_1+d_l}&1&e_1&\cdots&e_1^{k-l-1}\\
    \frac{e_2^{p}}{e_2+d_1}&\frac{e_2^{p+1}}{e_2+d_2}&\cdots&\frac{e_2^{p+l-1}}{e_2+d_l}&1&e_2&\cdots&e_2^{k-l-1}\\
    \vdots&\vdots&\ddots&\vdots&\vdots&\vdots&\ddots&\vdots\\
    \frac{e_k^{p}}{e_k+d_1}&\frac{e_k^{p+1}}{e_k+d_2}&\cdots&\frac{e_k^{p+l-1}}{e_k+d_l}&1&e_k&\cdots&e_k^{k-l-1}
  \end{array}\right),
\end{equation}
maintaining the assumptions $0\leq p, \, 0\leq l\leq k$ and $ p+l-1\leq k-l$.  

Let us now consider the $l$-th column of the matrix; we may write it as
\[
\left(
\begin{array}{c}
  \frac{e_1^{p+l-1}-(-d_l)^{p+l-1}+(-d_l)^{p+l-1}}{e_1+d_l}\\
  \frac{e_2^{p+l-1}-(-d_l)^{p+l-1}+(-d_l)^{p+l-1}}{e_2+d_l}\\
  \vdots\\
  \frac{e_k^{p+l-1}-(-d_l)^{p+l-1}+(-d_l)^{p+l-1}}{e_k+d_l}
\end{array}
\right)=
\left(
\begin{array}{c}
  \sum_{j=1}^{p+l-1}e_1^{p+l-1-j}(-d_l)^{j-1}+\frac{(-d_l)^{p+l-1}}{e_1+d_l}\\
  \sum_{j=1}^{p+l-1}e_2^{p+l-1-j}(-d_l)^{j-1}+\frac{(-d_l)^{p+l-1}}{e_2+d_l}\\
  \vdots\\
  \sum_{j=1}^{p+l-1}e_k^{p+l-1-j}(-d_l)^{j-1}+\frac{(-d_l)^{p+l-1}}{e_k+d_l}
\end{array}
\right).
\]
Since, thanks to the assumption $ p+l-1\leq k-l$, the first terms above are linear combinations of columns $l+1$  through $k$ we obtain
\[
\det(CV_k^{(l,p)}(\mathbf{e},\mathbf{d}))=\det\left(\begin{array}{ccccccccc}
    \frac{e_1^{p}}{e_1+d_1}&\frac{e_1^{p+1}}{e_1+d_2}&\cdots&\frac{(-d_l)^{p+l-1}}{e_1+d_l}&1&e_1&\cdots&e_1^{k-l-1}\\
    \frac{e_2^{p}}{e_2+d_1}&\frac{e_2^{p+1}}{e_2+d_2}&\cdots&\frac{(-d_l)^{p+l-1}}{e_2+d_l}&1&e_2&\cdots&e_2^{k-l-1}\\
    \vdots&\vdots&\ddots&\vdots&\vdots&\vdots&\ddots&\vdots\\
    \frac{e_k^{p}}{e_k+d_1}&\frac{e_k^{p+1}}{e_k+d_2}&\cdots&\frac{(-d_l)^{p+l-1}}{e_k+d_l}&1&e_k&\cdots&e_k^{k-l-1}
  \end{array}\right), 
\]
which after implementing similar operations for the remaining first $l-1$ columns 
leads to: 
\[
\det(CV_k^{(l,p)}(\mathbf{e},\mathbf{d}))=\det\left(\begin{array}{ccccccccc}
    \frac{(-d_1)^{p}}{e_1+d_1}&\frac{(-d_2)^{p+1}}{e_1+d_2}&\cdots&\frac{(-d_l)^{p+l-1}}{e_1+d_l}&1&e_1&\cdots&e_1^{k-l-1}\\
    \frac{(-d_1)^{p}}{e_2+d_1}&\frac{(-d_2)^{p+1}}{e_2+d_2}&\cdots&\frac{(-d_l)^{p+l-1}}{e_2+d_l}&1&e_2&\cdots&e_2^{k-l-1}\\
    \vdots&\vdots&\ddots&\vdots&\vdots&\vdots&\ddots&\vdots\\
    \frac{(-d_1)^{p}}{e_k+d_1}&\frac{(-d_2)^{p+1}}{e_k+d_2}&\cdots&\frac{(-d_l)^{p+l-1}}{e_k+d_l}&1&e_k&\cdots&e_k^{k-l-1}
  \end{array}\right).  
\]
Now it suffices to factor $(-d_1)^p, \cdots, (-d_l)^{p+l-1}$ in order to obtain 
a straightforward relation between the determinants of the modified Cauchy-Vandermonde matrix  and the Cauchy-Vandermonde matrix \eqref{eq:CV}, \eqref{eq:detCV}
\[
\det(CV_k^{(l,p)}(\mathbf{e},\mathbf{d}))=(-d_1)^{p}(-d_2)^{p+1}\cdots(-d_l)^{p+l-1}\cdot\det(CV_k^{l}(\mathbf{e},\mathbf{d})),
\]
from which, after restoring a general $\mathbf{a}$ which contributes the factor 
$\mathbf{a}_{[1,l]}$, the result follows.
\end{proof} 
In the final step of generalizing Cauchy-Vandermonde matrices we 
introduce a family of matrices of this type attached to a Stieltjes transform of a positive measure. 

~\begin{definition} \label{def:CSV}
 Given  a (strictly) positive vector $\mathbf{e}\in \R^k$, a non-negative 
 number $c$, an index $l$ such that $0\leq l\leq k$, another index $p$ such that $0\leq p, \, p+l-1\leq k-l$, and a positive measure $\nu$ 
 with support in $\R_+$, a \textit{ Cauchy-Stieltjes-Vandermonde (CSV) matrix} is that of the form 
\begin{equation}
  CSV_k^{(l,p)}(\mathbf{e},\nu, c)=\left(\begin{array}{cccccccc}
    e_1^{p}\hat\nu_c(e_1)&e_1^{p+1}\hat \nu_c(e_1)&\cdots&e_1^{p+l-1}\hat\nu_c(e_1)&1&e_1&\cdots&e_1^{k-l-1}\\
    e_2^{p}\hat\nu_c(e_2)&e_2^{p+1}\hat\nu_c(e_2)&\cdots&e_2^{p+l-1}\hat\nu_c(e_2)&1&e_2&\cdots&e_2^{k-l-1}\\
    \vdots&\vdots&\ddots&\vdots&\vdots&\vdots&\ddots&\vdots\\
    e_k^{p}\hat\nu_c(e_k)&e_k^{p+1}\hat\nu_c(e_k)&\cdots&e_k^{p+l-1}\hat\nu_c(e_k)&1&e_k&\cdots&e_k^{k-l-1}
  \end{array}\right),
\end{equation}
where $\hat \nu_c$ is the (shifted) Stieltjes transform of the measure 
$\nu$ and is given by
$
\hat \nu_c(y)=c+\int\frac{d\nu(x)}{y+x}. 
$

\end{definition} 
In the next theorem we establish explicit formulas 
for the determinant of the CSV matrix.  This theorem is essential for 
our solution of the interpolation problems \eqref{eq:Papprox}.  
\begin{theorem} \label{thm:detCSV}
Let $\nu$ be a positive measure with support in $\R_+$ and let $\mathbf{x}$ denote the vector ~$[x_1,x_2,\dots , x_l]\in ~\R^l $ and $d\nu^p(y)=y^p d\nu(y)$, respectively. 
Then 
\begin{enumerate} 
\item if either $c=0$ or $p+l-1<k-l$ then 
\begin{equation} \label{eq:detCSV1}
\det CSV_k^{(l,p)}(\mathbf{e},\nu, c)=(-1)^{lp+\frac{l(l-1)}{2}} \Delta_{[1,k]}(\mathbf{e})
\int_{0<x_1<x_2<\dots<x_l} \frac{\Delta_{[1,l]}(\mathbf{x})^2}{\Gamma_{[1,k], [1,l]}(\mathbf{e}; \mathbf{x})} d\nu^p(x_1)d\nu^p(x_2)\dots d\nu^p(x_l);
\end{equation} 
\item if $c>0$ and $p+l-1=k-l$ then 
\begin{equation} \label{eq:detCSV2}
\begin{split}
&\det CSV_k^{(l,p)}(\mathbf{e},\nu, c)=(-1)^{lp+\frac{l(l-1)}{2}} \Delta_{[1,k]}(\mathbf{e})\\&\cdot\Big(
\int_{0<x_1<x_2<\dots<x_l} \frac{\Delta_{[1,l]}(\mathbf{x})^2}{\Gamma_{[1,k], [1,l]}(\mathbf{e}; \mathbf{x})} d\nu^p(x_1)d\nu^p(x_2)\dots d\nu^p(x_l)\\
&+c
\int_{0<y_1<y_2<\dots<y_{l-1}} \frac{\Delta_{[1,l-1]}(\mathbf{y})^2}{\Gamma_{[1,k], [1,l-1]}(\mathbf{e}; \mathbf{y})} d\nu^p(y_1)d\nu^p(y_2)\dots d\nu^p(y_{l-1})\Big). 
\end{split}
\end{equation} 
\end{enumerate} 

\end{theorem} 
\begin{proof} 
Let us first consider the case $c=0$.  Using multilinearity of the determinant we obtain
\begin{equation*}
\det CSV_k^{(l,p)}(\mathbf{e},\nu, 0) =\int
\det \begin{pmatrix} \frac{e_1^p}{e_1+x_1}&\frac{e_1^{p+1}}{e_1+x_2}&\cdots&\frac{e_1^{p+l-1}}{e_1+x_l}&1&e_1&\dots&e_1^{k-l+1}\\
\frac{e_2^p}{e_2+x_1}&\frac{e_2^{p+1}}{e_2+x_2}&\cdots&\frac{e_2^{p+l-1}}{e_2+x_l}&1&e_2&\dots&e_2^{k-l+1}\\
\vdots&\vdots&\ddots&\vdots&\vdots&\vdots&\ddots&\vdots\\
\frac{e_k^p}{e_k+x_1}&\frac{e_k^{p+1}}{e_k+x_2}&\cdots&\frac{e_k^{p+l-1}}{e_k+x_l}&1&e_k&\dots&e_k^{k-l+1}\
\end{pmatrix} d\nu(x_1)d\nu(x_2)\cdots d\nu(x_l). 
\end{equation*}
Then by Theorem \ref{thm:mCV} 
\begin{equation*}
\det CSV_k^{(l,p)}(\mathbf{e},\nu, 0)=(-1)^{lp+\frac{l(l-1)}{2}}
\int (x_1x_2\cdots x_l)^p x_1^0x_2^1\cdots x_l^{l-1} 
\frac{\Delta_{[1,k]}(\mathbf{e})\Delta_{[1,l]}(\mathbf{x})}{\Gamma_{[1,k], [1,l]}(\mathbf{e};\mathbf{x})}d\nu(x_1)d\nu(x_2)\cdots d\nu(x_l).  
\end{equation*}
Let us now consider the action of the group of permutations on $l$ letters, denoted 
$S_l$, on individual terms of the integrand.  The product measure 
is invariant under the action and so are $(x_1x_2\cdots x_l)^p$ and $\Gamma_{[1,n], [1,l]}(\mathbf{e};\mathbf{x})$.  Let $\sigma.\mathbf{x}=[x_{\sigma(1)}, x_{\sigma(2)}, \cdots, x_{\sigma(l)}]$ then 
\begin{align*}
&\det CSV_k^{(l,p)}(\mathbf{e},\nu,0)\\
&=\frac{(-1)^{lp+\frac{l(l-1)}{2}}\Delta_{[1,k]}(\mathbf{e})}{l!}
\int \frac{(x_1x_2\cdots x_l)^p }{\Gamma_{[1,k], [1,l]}(\mathbf{e};\mathbf{x})}\big(\sum_{\sigma\in S_l} x_{\sigma(1)}^0x_{\sigma(2)}^1\cdots x_{\sigma(l)} ^{l-1}\Delta_{[1,l]}(\sigma.\mathbf{x})\big)\, d\nu(x_1)d\nu(x_2)\cdots d\nu(x_l)\\
&= \frac{(-1)^{lp+\frac{l(l-1)}{2}}\Delta_{[1,k]}(\mathbf{e})}{l!}
\int \frac{(x_1x_2\cdots x_l)^p \Delta^2_{[1,l]}(\mathbf{x})}{\Gamma_{[1,k], [1,l]}(\mathbf{e};\mathbf{x})}\, d\nu(x_1)d\nu(x_2)\cdots d\nu(x_l),  
\end{align*}
where in the last step we used $\Delta_{[1,l]}(\sigma.\mathbf{x})=\sgn(\sigma)
\Delta_{[1,l]}(\mathbf{x})$.  Since now the integrand is invariant under 
the action of $S_l$ we integrate over $x_1<x_2<\cdots<x_l$, multiply by $l!$, and 
restrict integration to $R_+$ in view of the condition on the support of $\nu$, obtaining the final formula for this case.  

The next case is $0\leq c$ but $p+l-1<k-l$.  In this case 
every column $j, 1\leq j\leq l$ is a sum: 
$$
c\begin{pmatrix}e_1^{p+j-1}\\e_2^{p+j-1}\\\vdots\\e_k^{p+j-1}\end{pmatrix}
+\begin{pmatrix}e_1^{p+j-1}\nu_0(e_1)\\e_2^{p+j-1}\nu_0(e_2)\\\vdots\\e_k^{p+j-1}\nu_0(e_k)\end{pmatrix}, 
$$
and the condition $p+l-1<k-l$ ensures that the first vector, namely the one multiplied by $c$, appears 
in the Vandermonde part of the matrix, hence by antisymmetry of the determinant implying that this case 
reduces to the case $c=0$.  

Finally, the last case $p+l-1=k-l$ can handled in a similar fashion, except that 
now in the $l$th column we have a term which does not appear in the original
Cauchy part: 
\begin{equation*}
\begin{split}
&\det CSV_k^{(l,p)}(\mathbf{e},\nu,c)\\
&=\det\left(\begin{array}{cccccccc}
    e_1^{p}\hat\nu_0(e_1)&e_1^{p+1}\hat \nu_0(e_1)&\cdots&ce_1^{p+l-1}+e_1^{p+l-1}\hat\nu_0(e_1)&1&e_1&\cdots&e_1^{k-l-1}\\
    e_2^{p}\hat\nu_0(e_2)&e_2^{p+1}\hat\nu_0(e_2)&\cdots&ce_2^{p+l-1}+e_2^{p+l-1}\hat\nu_0(e_2)&1&e_2&\cdots&e_2^{k-l-1}\\
    \vdots&\vdots&\ddots&\vdots&\vdots&\vdots&\ddots&\vdots\\
    e_k^{p}\hat\nu_0(e_k)&e_k^{p+1}\hat\nu_0(e_k)&\cdots&ce_k^{p+l-1}+e_k^{p+l-1}\hat\nu_0(e_k)&1&e_k&\cdots&e_k^{k-l-1}
  \end{array}\right). 
  \end{split}
\end{equation*}
It suffices now to split the determinant into two, then move the term involving $c$ to the Vandermonde 
part, effectively lowering $l$ to $l-1$ for this term and then apply the same 
method as in the proof of the $c=0$ case.  
\end{proof} 
Our goal in the remainder of this section is to connect 
the determinantal formalism we have developed to the 
interpolation problem \eqref{eq:Papprox} (see Theorem \ref{thm:solPapproxbis}).  
To this end we set (see \eqref{eq:Wc} and Theorem \ref{thm:solPapproxbis}): 
\begin{equation*}
e_j=-z_j=\frac{1}{m_{j'}^2}, \qquad \nu =\mu, \quad 1\leq j\leq n, 
\end{equation*}
and observe that $W(z_j)=\hat \mu_c(e_j)$ by \eqref{eq:Wc}.  
\begin{theorem}\label{eq:evalDkEk}
Let $D_k, E_k$ be the determinants defined in Theorem \ref{thm:solPapproxbis}.  
Then 
\begin{equation*}
D_k=(-1)^{\lfloor \frac k2\rfloor \lfloor \frac {k+1}{2} \rfloor} \det CSV_k^{({\lfloor \frac k2}\rfloor,1)}(\mathbf{e},\mu,c), \qquad E_k=(-1)^{k+\lfloor  \frac{k}{2} \rfloor \lfloor \frac {k+1}{2 }\rfloor} \mathbf{e}_{[1,k]}\det CSV_k^{(\lfloor \frac{k+1}{2}\rfloor,0)}(\mathbf{e},\mu,c). 
\end{equation*}
\end{theorem} 
\begin{proof}
This is a straightforward computation requiring only to factor 
$(-1)^j$ from any column containing $z_i^j=(-e_i)^j$ and, in the case of 
$E_k$, we also need to factor $(z_1z_2\cdots z_k)$ and reshuffle the columns 
to bring the matrix to the CSV form.  
\end{proof}
Now, it suffices to use theorem \ref{thm:detCSV} to conclude that both 
$D_k$ and $E_k$ are nonzero, provided that all $e_j$ are distinct (
which is the same as our nondegeracy conditions on the masses $m_j$). 
\begin{corollary} \label{cor:DkEknonzero}
The determinants $D_k$ and $E_k$ appearing in the solution 
to the interpolation problem stated in Theorem \ref{thm:solPapproxbis} are non zero for any $k$, 
$1\leq k\leq n$.  
\end{corollary} 

We finish this subsection by giving a complete 
solution to the inverse problem in terms of determinants of CSV matrices.  
To lessen the burden of keeping track of signs resulting from 
manipulations of matrix columns needed to bring matrices to the CSV form 
we opt for the display using the absolute value of determinants to 
quickly and compactly present the formulas.  
Thus, till further notice, we will denote 
\begin{equation}\label{eq:calD}
\D_k^{(l,p)}=\abs{\det CSV_k^{(l,p)}(\mathbf{e},\mu,c)}, 
\end{equation}
with the proviso that the arguments $\mathbf{e}, \mu, c$ are fixed.  
With this notation in place the formulas of Theorem \ref{eq:evalDkEk} take 
the form: 
\begin{equation}\label{eq:absDkEk}
\abs{D_k}=\D_k^{(\floor{\frac k2}, 1)}, \quad \abs{E_k}=\mathbf{e}_{[1,k]} \D_k^{(\floor{\frac {k+1}{2}}, 0)}, 
\end{equation}
from which one obtains determinantal formulas for 
the coefficients of the polynomials $\hat q_k, \hat p_k, \hat Q_k, \hat P_k$ leading 
via equations \eqref{eq:inversegodd}, \eqref{eq:inversegeven} to a complete 
solution of the inverse problem \ref{def:ISP}.  
\begin{theorem}\label{thm:detsolISP}
Suppose the Weyl function $W(z)$ is given by \eqref{eq:Wc} along with 
positive distinct constants (masses) $m_1, m_2, \cdots, m_n$.  Then 
there exists a unique solution to the inverse problem specified in Definition \ref{def:ISP}: 
\begin{subequations}
\begin{align}
&&g_{k'}&=\frac{\D_k^{(\frac{k-1}{2},1)}\D_{k-1}^{(\frac{k-1}{2},1)}}
{\mathbf{e}_{[1,k]}\D_k^{(\frac{k+1}{2},0)}\D_{k-1}^{(\frac{k-1}{2},0)}},   &\text{ if k  is odd},  \label{eq:detinversegodd}\\
&&g_{k'}&= \frac{\D_k^{(\frac{k}{2},1)}\D_{k-1}^{(\frac k2 -1,1)}}
{\mathbf{e}_{[1,k]}\D_k^{(\frac k2,0)}\D_{k-1}^{(\frac{k}{2},0)}}, &\text{ if k is even}. \label{eq:detinversegeven}
\end{align}
\end{subequations}
Likewise, 
\begin{subequations}
\begin{align}
h_{k'}&=\frac
{\mathbf{e}_{[1,k-1]}\D_k^{(\frac{k+1}{2},0)}\D_{k-1}^{(\frac{k-1}{2},0)}}{\D_k^{(\frac{k-1}{2},1)}\D_{k-1}^{(\frac{k-1}{2},1)}},   &\text{ if k is odd},  \label{eq:detinversehodd}\\
h_{k'}&= \frac
{\mathbf{e}_{[1,k-1]}\D_k^{(\frac k2,0)}\D_{k-1}^{(\frac{k}{2},0)}}{\D_k^{(\frac{k}{2},1)}\D_{k-1}^{(\frac k2 -1,1)}}, &\text{ if k is even}. \label{eq:detinversegeven}
\end{align}
\end{subequations}
\end{theorem} 
\begin{proof} 
The formulas follow from equations \eqref{eq:inversegodd} and \eqref{eq:inversegeven}, as well as Theorem \ref{thm:solPapproxbis}.  The 
question of signs involved in the identification of the CSV determinants is 
addressed by taking the absolute values in all formulas needed to produce 
positive outcomes $g_{k'}$.  The formulas for $h_{k'}$ follow from 
the relation $g_{k'}h_{k'}=m_{k'}^2$.  
\end{proof} 
Finally, recalling that the original peakon problem \eqref{mCH_ode} was formulated 
in the $x$ space, using the relation $h_j=m_je^{x_j}$ (see equation \eqref{eq:xLax}), 
we arrive at the inverse formulae 
relating the spectral data and the positions of peakons given by $x_j$.  
\begin{theorem} \label{thm:inversex}
Given positive and distinct constants $m_j$, let $\Phi$ be the solution to the boundary value 
problem \ref{eq:xLaxBVP} with associated spectral data $\{d\mu, c\}$.  
Then the positions $x_j$ (of peakons) in the discrete measure 
$m=2\sum_{j=1}^n m_j \delta_{x_j} $ can be expressed in 
terms of the spectral data as: 

\begin{subequations}
\begin{align}
&&x_{k'}&=\ln \frac
{\mathbf{e}_{[1,k-1]}\D_k^{(\frac{k+1}{2},0)}\D_{k-1}^{(\frac{k-1}{2},0)}}{m_{k'}\D_k^{(\frac{k-1}{2},1)}\D_{k-1}^{(\frac{k-1}{2},1)}},   &\text{ if k  is odd}, \label{eq:detinversexodd}\\
&&x_{k'}&= \ln \frac
{\mathbf{e}_{[1,k-1]}\D_k^{(\frac k2,0)}\D_{k-1}^{(\frac{k}{2},0)}}{m_{k'}\D_k^{(\frac{k}{2},1)}\D_{k-1}^{(\frac k2 -1,1)}}, &\text{ if k is even}, \label{eq:detinversexeven}
\end{align}
\end{subequations}
with $\D_k^{(l,p)}$ defined in \eqref{eq:calD}, $k'=n-k+1, \, 1\leq k\leq n$ 
and the convention that $\D_0^{l,p}=1$.  
\end{theorem} 

Finally, we can relax the condition that the masses be distinct, observing that the Vandermonde determinants $\Delta_{[1, r]}(\mathbf{e}), \, r=k, k-1$, 
cancel out in all expressions of the type
\begin{equation*}
\frac{\D_k^{( l_1,p_1)}\D_{k-1}^{(l_2,p_2)}}{\D_k^{(l_3,p_3)}\D_{k-1} ^{(l_4,p_4)}}
\end{equation*}
as can be seen from the determinantal expressions 
given in Theorem \ref{thm:detCSV} (see \eqref{eq:detCSV1} and \eqref{eq:detCSV2}).  
\begin{theorem} \label{thm:minversex}
Given {\bf positive} constants $m_j$, let $\Phi$ be the solution to the boundary value 
problem \ref{eq:xLaxBVP} with associated spectral data $\{d\mu, c\}$.  
Then the positions $x_j$ (of peakons) in the discrete measure 
$m=2\sum_{j=1}^n m_j \delta_{x_j} $ can be expressed in 
terms of the spectral data as given by the continuous extension of the formulas \eqref{eq:detinversexodd} and \eqref{eq:detinversexeven} to all, including 
coinciding, masses.  
\end{theorem}
\begin{proof} 
We give a short proof of this statement.  The forward map $\mathcal{S}:
(m_1,m_2,\cdots, m_n; x_1,x_2,\cdots, x_n)\rightarrow (m_1,m_2,\cdots, \mathbf{\zeta}, \mathbf{b}, c)$ and its inverse $\mathcal{S}^{-1}$, 
which exists by Theorem \ref{thm:contInverse}, are continuous.  
The formulas \eqref{eq:detinversexodd} and \eqref{eq:detinversexeven}
were originally defined for distinct masses but, after cancellation of 
the Vandermonde determinants mentioned above, have continuous extensions 
to all positive masses, distinct or not.  By uniqueness and continuity of $\mathcal{S}$ the extended formulas are then the formulas valid for all positive masses.  

\end{proof}

\section{Multipeakons for $n=2K$}\label{sec:evenpeakons} 
Even though the only difference between even and odd $n$ is that $c=0$ when $n$ is even (see Theorem \ref{thm:W}),  and $c>0$ otherwise, we nevertheless present these two cases 
separately in an attempt to underscore subtle differences in the asymptotic 
behaviour of peakons for these two cases.   
\subsection{Closed formulae for $n=2K$}
If we assume that $x_1(0)<x_2(0)<\cdots<x_{2K}(0)$ then this condition will hold at least in a small interval containing $t=0$.  Thus by Theorem \ref{thm:inversex} we have the following result.

\begin{theorem}\label{thm:peakon_even}
Assuming the notation of Theorem \ref{thm:inversex}, the mCH equation \eqref{eq:m1CH} with the regularization 
of the singular term $u_x^2 m$ given by $\avg{u_x^2}m$ admits the multipeakon solution
\begin{equation}\label{eq:umultipeakoneven}
u(x,t)=\sum_{k=1}^{2K}m_{k'}(t)\exp(-|x-x_{k'}(t)|),
\end{equation}
where $x_{k'}$ are given by equations \eqref{eq:detinversexodd} and \eqref{eq:detinversexeven}, 
with the peakon spectral measure 
\begin{equation}\label{eq:peakon sm}
d\mu=\sum_{j=1}^{K} b_j(t) \delta_{\zeta_j}, 
\end{equation} 
$b_j(t)=b_j(0)e^{\frac{2t}{\zeta_j}}, \, 0<b_j(0)$, ordered eigenvalues $0<\zeta_1<\cdots<\zeta_K$ and $c=0$ in \eqref{eq:calD}.
\end{theorem}
\begin{proof} 
We only need to discuss the time evolution of the spectral measure.  To this 
end we recall that the Weyl function $W(z)$ defined in \eqref{eq:defWeyl} 
undergoes the time evolution  dictated by \eqref{eq:tderqp}; 
a simple computation gives 
\begin{equation*}
\dot W=\frac 2z W-\frac{2L}{z}, 
\end{equation*}
which, in turn, implies
$\dot b_j=\frac{2}{\zeta_j} b_j, 1\leq j\leq K$ by virtue of 
Corollary \ref{cor:spectrum}.  The remaining statement regarding 
the multipeakon solutions follows from our solution of the inverse problem 
and the fact that to formulate the time evolution we used the distributionally compatible Lax pair (see Apendix \ref{lax_mch}).  
\end{proof} 

Even though one could easily give examples of peakon solutions 
based directly on Theorem \ref{thm:peakon_even} it is helpful 
to examine the explicit formulas for the evaluation of 
CSV determinants presented in Theorem \ref{thm:detCSV} (see equation \eqref{eq:calD} for notation), adjusted to the case of even $n$.  

\begin{theorem} \label{thm:Dklm-peakon}  Let $n=2K, \, 0\leq l\leq K, \, 0\leq p, \,  p+l-1\leq k-l, \,1\leq k\leq 2K$ and let the peakon spectral measure be given by \eqref{eq:peakon sm}.  
Then 
\begin{enumerate} 
\item 
\begin{equation}\label{eq:Dklp} 
\D_k^{(l,p)}=\abs{\Delta_{[1,k]}(\mathbf{e})}\sum_{I\in\binom{[1,K]}{l} }
\frac{\Delta^2_I(\mathbf{\zeta})\mathbf{b}_{I} \mathbf{\zeta}^p_I}{\Gamma_{[1,k],I}(\mathbf{e};\mathbf{\zeta})}; 
\end{equation} 
\item in the asymptotic region $t\rightarrow + \infty$
\begin{equation} \label{eq:Dklp+}
\D_k^{(l,p)}=\abs{\Delta_{[1,k]}(\mathbf{e})}
\frac{\Delta^2_{[1,l]} (\mathbf{\zeta})\mathbf{b}_{[1,l]} \mathbf{\zeta}^p_{[1,l]}}{\Gamma_{[1,k],[1,l]}(\mathbf{e};\mathbf{\zeta})}\Big[1+\mathcal{O}(e^{-\alpha t})\Big], \quad 0< \alpha; 
\end{equation}
\item in the asymptotic region $t \rightarrow -\infty$
\begin{equation}\label{eq:Dklp-}
\D_k^{(l,p)}=\abs{\Delta_{[1,k]}(\mathbf{e})}
\frac{\Delta^2_{[1,l]^*} (\mathbf{\zeta})\mathbf{b}_{[1,l]^*} \mathbf{\zeta}^p_{[1,l]^*}}{\Gamma_{[1,k],[1,l]^*}(\mathbf{e};\mathbf{\zeta})}\Big[1+\mathcal{O}(e^{\beta t})\Big]
, \quad 0<\beta, 
\end{equation} 
where $[1,l]^*=[l^*=K-l+1,1^*=K]$ $(\textrm{reflection of the interval }[1,K])$.  
\end{enumerate} 
\end{theorem} 
\begin{proof} 
Equation \eqref{eq:Dklp} follows from Theorem \ref{thm:detCSV}, in particular 
equation \eqref{eq:detCSV1}, by taking $d\nu=d\mu$ there, and carrying 
out integration.  
The formula \eqref{eq:Dklp+} can be obtained from \eqref{eq:Dklp} by observing 
that the time dependence is confined to terms $\mathbf{b}_{I}= 
b_{j_1}(0)e^{\frac{2t}{\zeta_{j_1}}}b_{j_2}(0)e^{\frac{2t}{\zeta_{j_2}}}\cdots 
b_{j_l}(0)e^{\frac{2t}{\zeta_{j_l}}}$ of which the term with the smallest $l$-tuple of 
eigenvalues is dominant; the rest then follows from our ordering of 
the eigenvalues.  Finally, 
formula \eqref{eq:Dklp-} follows from a similar argument, except that 
for $t\rightarrow -\infty$ the dominant term corresponds to the largest $l$-tuple 
of eigenvalues.  

\end{proof} 
Before we display examples of formulas for multipeakons in the case of 
$n=2K$ we remind the reader that  
$e_j=\frac{1}{m^2_{j'}}, \, j'=2K-j+1$.  All examples below 
are derived following the same pattern: one takes formulas \eqref{eq:detinversexodd}, 
\eqref{eq:detinversexeven} and uses \eqref{eq:Dklp} to derive explicit expressions 
for positions $x_1, \cdots, x_{2K}$.   
\begin{example}[2-peakon solution; K=1]
  \begin{align*}
     &x_1=\ln\left(\frac{b_1}{\zeta_1m_1(1+\zeta_1m_2^2)}\right),\ \ \
     &x_2=\ln\left(\frac{b_1m_2}{1+\zeta_1m_2^2}\right).
       \end{align*}
\end{example}

\begin{example}[4-peakon solution; K=2]
  \begin{align*}
     &x_1=\ln\left(\frac{1}{m_1}\cdot\frac{b_1b_2(\zeta_2-\zeta_1)^2}{\zeta_1\zeta_2\left(b_1\zeta_1(1+\zeta_2m_2^2)(1+\zeta_2m_3^2)(1+\zeta_2m_4^2)+b_2\zeta_2(1+\zeta_1m_2^2)(1+\zeta_1m_3^2)(1+\zeta_1m_4^2)\right)}\right),&\\
     &x_2=\ln\left(m_2\cdot\frac{b_1b_2(\zeta_2-\zeta_1)^2\left(b_1(1+\zeta_2m_3^2)(1+\zeta_2m_4^2)+b_2(1+\zeta_1m_3^2)(1+\zeta_1m_4^2)\right)}{\left(b_1\zeta_1(1+\zeta_2m_3^2)(1+\zeta_2m_4^2)+b_2\zeta_2(1+\zeta_1m_3^2)(1+\zeta_1m_4^2)\right)}\right.&\\
     &\qquad\qquad \cdot\left.\frac{1}{\left(b_1\zeta_1(1+\zeta_2m_2^2)(1+\zeta_2m_3^2)(1+\zeta_2m_4^2)+b_2\zeta_2(1+\zeta_1m_2^2)(1+\zeta_1m_3^2)(1+\zeta_1m_4^2)\right)}\right),&\\
     &x_3=\ln\left(\frac{1}{m_3}\cdot\frac{\left(b_1(1+\zeta_2m_4^2)+b_2(1+\zeta_1m_4^2)\right)\left(b_1(1+\zeta_2m_3^2)(1+\zeta_2m_4^2)+b_2(1+\zeta_1m_3^2)(1+\zeta_1m_4^2)\right)}{(1+\zeta_1m_4^2)(1+\zeta_2m_4^2)\left(b_1\zeta_1(1+\zeta_2m_3^2)(1+\zeta_2m_4^2)+b_2\zeta_2(1+\zeta_1m_3^2)(1+\zeta_1m_4^2)\right)}\right),&\\
     &x_4=\ln\left(m_4\cdot\frac{b_1(1+\zeta_2m_4^2)+b_2(1+\zeta_1m_4^2)}{(1+\zeta_1m_4^2)(1+\zeta_2m_4^2)}\right).&
       \end{align*}
       \end{example}
       \begin{example}[a general formula for the last position $x_{2K}=x_{1'}$]
\begin{equation*}
x_{2K}=x_{1'}=\ln\frac{\D_1^{(1,0)}}{m_{1'}}=\ln \frac{\hat \mu_0(e_1)}{m_{2K}}=
\ln\frac{\sum_{i=1}^K \frac{b_i}{\frac{1}{m_{2K}^2}+\zeta_i}}{m_{2K}}=\ln m_{2K} \sum_{i=1}^K \frac{b_i}{1+m_{2K}^2\zeta_i}. 
\end{equation*}
\end{example}

\subsection{Global existence for $n=2K$}
 As time varies, the initial order $x_1(0)<x_2(0)<\cdots<x_{2K}(0)$ might cease to hold. In this subsection, we formulate a sufficient  condition needed to 
ensure that the peakon flow exists globally in time.
\begin{theorem}\label{thm:global}
  Given arbitrary spectral data $$\{b_j>0, \, 0<\zeta_1<\zeta_2<\cdots< \zeta_K: 1\leq j \leq K \}, $$ suppose the masses $m_k$ satisfy
\begin{subequations}
\begin{align}
&&\frac{\zeta_K^{\frac{k-1}{2}}}{\zeta_1^{\frac{k+1}{2}}}&< m_{(k+1)' }m_{k'}, \qquad &\text { for all odd }k, \qquad   1\leq k\leq 2K-1, \\
&&\frac{m_{(k+2)'}m_{(k+1)'}}{(1+m_{(k+1)'}^2\zeta_1)(1+m_{(k+2)'}^2\zeta_1)}&<
\frac{\zeta_1^{\frac{k+1}{2}}}{\zeta_K^{\frac{k-1}{2}}}
\frac{2\, \textnormal{min}_j(\zeta_{j+1}-\zeta_j)^{k-1}}{(k+1)(\zeta_K-\zeta_1)^{k+1}}, 
\quad &\text{ for all odd } k, \qquad 1\leq k\leq 2K-3. \label{eq:seccond}
\end{align}
\end{subequations}
Then the positions obtained from inverse formulas \eqref{eq:detinversexodd}, 
\eqref{eq:detinversexeven} are ordered $x_1<x_2<\cdots<x_{2K}$ and the multipeakon solutions \eqref{eq:umultipeakoneven} exist for arbitrary $t\in R$.
\end{theorem}
\begin{proof}
The solutions described in Theorem \ref{thm:inversex} are valid peakon solutions as long as $x_1<x_2<\cdots<x_{2K}$ holds.  We write these conditions as:

\begin{subequations}
\begin{align}
&&x_{(k+1)'}&<x_{k'}, & \text { for all odd }k, \qquad  1\leq k\leq 2K-1, \label{eq:order1}\\
&&x_{(k+2)'}&<x_{(k+1)'}, & \text{ for all odd }k, \qquad  1\leq k\leq 2K-3, \label{eq:order2}
\end{align}
\end{subequations}

and use equations \eqref{eq:detinversexodd}, \eqref{eq:detinversexeven} to 
obtain equivalent conditions
\begin{subequations}
\begin{align}
\frac{1}{m_{(k+1)'}m_{k'}}<\frac{\D_{k+1}^{(\frac{k+1}{2},1)} \D_{k-1}^{(\frac{k-1}{2},0)}}{\D_{k+1}^{(\frac{k+1}{2},0)} \D_{k-1}^{(\frac{k-1}{2},1)}},&& \text { for all odd }k, \qquad   1\leq k\leq 2K-1, \label{eq:order1bis} \\
\frac{1}{m_{(k+2)'}m_{(k+1)'}}<\frac{\D_{k+2}^{(\frac{k+1}{2},1)} \D_{k}^{(\frac{k+1}{2},0)}}{\D_{k+2}^{(\frac{k+3}{2},0)} \D_{k}^{(\frac{k-1}{2},1)}},&& \text{ for all odd } k, \qquad 1\leq k\leq 2K-3.  \label{eq:order2bis} 
\end{align} 
\end{subequations}
 
Note that equation \eqref{eq:Dklp} implies easily that the inequality 
\begin{equation}\label{eq:order1step3}
\frac{\D_{k+1}^{(\frac{k+1}{2},1)} \D_{k-1}^{(\frac{k-1}{2},0)}}{\D_{k+1}^{(\frac{k+1}{2},0)} \D_{k-1}^{(\frac{k-1}{2},1)}}>\frac{\zeta_1^{\frac{k+1}{2}}}{\zeta_K^{\frac{k-1}{2}}}
\end{equation}
holds uniformly in $t$ (we recall that the coefficients $b_j$ depend on $t$). Thus 
if we impose 
\begin{equation*}
\frac{1}{m_{(k+1)'}m_{k'}}<\frac{\zeta_1^{\frac{k+1}{2}}}{\zeta_K^{\frac{k-1}{2}}}\text { for all odd }k, \qquad   k\leq 2K-1,  
\end{equation*} 
then equations \eqref{eq:order1} hold automatically.  

Now we turn to the second inequality, namely \eqref{eq:order2bis}, which is needed whenever $K\geq 2$. 
It is convenient to consider a slightly more general expression, namely, 
\begin{equation*} 
\frac{\D_{k+2}^{(l,1)} \D_{k}^{(l,0)}}{\D_{k+2}^{(l+1,0)} \D_{k}^{(l-1,1)}}, \qquad 1\leq l \leq K-1, 
\end{equation*}
for which after using 
\eqref{eq:Dklp} we obtain the inequality
\begin{equation} \label{eq:order2step3}
\frac{\D_{k+2}^{(l,1)} \D_{k}^{(l,0)}}{\D_{k+2}^{(l+1,0)} \D_{k}^{(l-1,1)}}> \frac{\zeta_1^{l}}{\zeta_K^{l-1}}
\frac{\displaystyle\sum_{A,B\in \binom{[1,K]}{l }}\frac{\Delta_A^2 \Delta_B^2 \mathbf{b}_A\mathbf{b}_B}{\Gamma_{[1,k+2],A}\Gamma_{[1,k],B}}}{\displaystyle
\sum_{I\in \binom{[1,K]}{l-1}, J\in \binom{[1,K]} {l+1}}
\frac{\Delta^2_I \Delta^2_J \mathbf{b}_I\mathbf{b}_J}{\Gamma_{[1,k],I} \Gamma_{[1,k+2],J}}}, 
\end{equation}

where, to ease off notation, we temporarily suspended displaying the dependence on $\mathbf{\zeta}, \mathbf{e}$. 

We focus now on rewriting the denominator of the above expression.  First, we note that 
since the cardinality of $J$ exceeds that of $I$ it is always possible to find a unique smallest  index $i$ in $J$ which is not in $I$.  This leads to the map:
\begin{align}\label{eq:Map}
\Phi: \binom{[1,K]}{l-1}&\times \binom{[1,K]}{l+1}&\longrightarrow  &&\binom{[1,K]}{l}\times \binom{[1,K]}{l},\\
&(I,J):&\longmapsto &&(A=I\cup \{i\}, \, B=J\setminus \{i\}), \quad \text{ for all } l, \qquad l\leq K-1.  \notag
\end{align}

For $A=I\cup \{i\}, \, B=J\setminus \{i\}$ we clearly have $A,B \in \binom{[1,K]}{l}$, $\mathbf{b}_I\mathbf{b}_J=\mathbf{b}_A\mathbf{b}_B$, as well as  
\begin{align*}
&\frac{\Delta_I^2\Delta_J^2}{\Gamma_{[1,k],I}\Gamma_{[1,k+2],J}}&\\
=&\frac{\Delta_{\{i\}, B}^2}{\Delta_{\{i\}, I}^2}\cdot\frac{1}{(e_{k+1}+\zeta_i)(e_{k+2}+\zeta_i)}\cdot\frac{\Delta_A^2\Delta_B^2}{\Gamma_{[1,k],A}\Gamma_{[1,k+2],B}}&\\
\leq&\frac{(\zeta_K-\zeta_1)^{2l}}{\textnormal{min}_j(\zeta_{j+1}-\zeta_j)^{2(l-1)}}\cdot\frac{1}{(e_{k+1}+\zeta_1)(e_{k+2}+\zeta_1)}\frac{\Delta_A^2\Delta_B^2}{\Gamma_{[1,k],A}\Gamma_{[1,k+2],B}}
\end{align*}
which implies

\begin{align*}
&{\displaystyle
\sum_{I\in \binom{[1,K]}{l-1}, J\in \binom{[1,K]} {l+1}}
\frac{\Delta^2_I \Delta^2_J \mathbf{b}_I\mathbf{b}_J}{\Gamma_{[1,k],I} \Gamma_{[1,k+2],J}}}&\\
\leq&\frac{(\zeta_K-\zeta_1)^{2l}}{\textnormal{min}_j(\zeta_{j+1}-\zeta_j)^{2(l-1)}}\cdot\frac{1}{(e_{k+1}+\zeta_1)(e_{k+2}+\zeta_1)}&\\
&\qquad \cdot\sum_{\substack{A,B \in\binom{[1,K]}{l}\\
(A,B)\in \textnormal{Image}(\Phi)
}}\#[\Phi^{-1}(A,B)]\frac{\Delta_A^2\Delta_B^2\mathbf{b}_A \mathbf{b}_B}{\Gamma_{[1,k],A}\Gamma_{[1,k+2],B}},&
\end{align*}


where $\#[\Phi^{-1}(A,B)]$ counts the number of pairs $(I,J)$ which are mapped by $\Phi$ into the same $(A,B)$.
However,  by construction, $\#[\Phi^{-1}(A,B)]\leq 1$ for $l=1$, while for $1<l$ 
two distinct pairs $(I_1,J_1)\neq(I_2,J_2)$ are mapped to the same $(A,B)$ if,
for the smallest $i_1\in J_1\setminus I_1$ and  the smallest  $i_2\in J_2\setminus I_2$, there exists $L\in \binom{[1,K]}{l-2}, M\in \binom{[1,K]}{l}$ such that
\begin{align*}
&I_1=L\cup \{i_2\}, &J_1=M\cup\{i_1\}, \\&I_2=L\cup\{i_1\}, &J_2=M\cup\{i_2\},
\end{align*}
in which case $A=L\cup \{i_1\}\cup \{i_2\}, \, B=M$.  Thus $\#[\Phi^{-1}(A,B)]$ is bounded from above
by the number of ways we can select an individual entry from $A$, since 
once $i_1$ is selected so is $i_2$ and $M$,  hence $\#[\Phi^{-1}(A,B)] \leq l$, and

\begin{align}
&{\displaystyle
\sum_{I\in \binom{[1,K]}{l-1}, J\in \binom{[1,K]} {l+1}}
\frac{\Delta^2_I \Delta^2_J \mathbf{b}_I\mathbf{b}_J}{\Gamma_{[1,k],I} \Gamma_{[1,k+2],J}}}\\
\leq&l\, \frac{(\zeta_K-\zeta_1)^{2l}}{ \textnormal{min}_j(\zeta_{j+1}-\zeta_j)^{2(l-1)}}\cdot\frac{1}{(e_{k+1}+\zeta_1)(e_{k+2}+\zeta_1)}\cdot\sum_{\substack{A,B \in\binom{[1,K]}{l}\\
(A,B)\in \textnormal{Image}(\Phi)
}}\frac{\Delta_A^2\Delta_B^2 \mathbf{b}_A \mathbf{b}_B}{\Gamma_{[1,k],A}\Gamma_{[1,k+2],B}}&\nonumber\\
\leq&l\,\frac{(\zeta_K-\zeta_1)^{2l}}{\textnormal{min}_j(\zeta_{j+1}-\zeta_j)^{2(l-1)}}\cdot\frac{1}{(e_{k+1}+\zeta_1)(e_{k+2}+\zeta_1)}\cdot\sum_{\substack{A,B \in\binom{[1,K]}{l}
}}\frac{\Delta_A^2\Delta_B^2 \mathbf{b}_A \mathbf{b}_B}{\Gamma_{[1,k],A}\Gamma_{[1,k+2],B}},&\label{inequl}
\end{align}

which, upon substituting into \eqref{eq:order2step3}, proves the bound
\begin{equation} \label{eq:order2step4}
\begin{split} 
\frac{\D_{k+2}^{(l,1)} \D_{k}^{(l,0)}}{\D_{k+2}^{(l+1,0)} \D_{k}^{(l-1,1)}}> \frac{\zeta_1^{l}}{\zeta_K^{l-1}}
\frac{\textnormal{min}_j(\zeta_{j+1}-\zeta_j)^{2(l-1)}}{l(\zeta_K-\zeta_1)^{2l}}\cdot(e_{k+1}+\zeta_1)(e_{k+2}+\zeta_1)\\=\frac{\zeta_1^{l}}{\zeta_K^{l-1}}
\frac{\textnormal{min}_j(\zeta_{j+1}-\zeta_j)^{2(l-1)}}{l(\zeta_K-\zeta_1)^{2l}}\cdot\frac{(1+m_{(k+1)'}^2\zeta_1)(1+m_{(k+2)'}^2\zeta_1)}{m^2_{(k+1)'}m^2_{(k+2)'}}.   
\end{split}
\end{equation}
Hence, setting $l=\frac{k+1}{2}$, if one takes 
\begin{equation*}
\frac{1}{m_{(k+2)'}m_{(k+1)'}} < \frac{\zeta_1^{\frac{k+1}{2}}}{\zeta_K^{\frac{k-1}{2}}}
\frac{2\, \textnormal{min}_j(\zeta_{j+1}-\zeta_j)^{k-1}}{(k+1)(\zeta_K-\zeta_1)^{k+1}}\cdot\frac{(1+m_{(k+1)'}^2\zeta_1)(1+m_{(k+2)'}^2\zeta_1)}{m^2_{(k+1)'}m^2_{(k+2)'}},   
\end{equation*}
then \eqref{eq:order2bis} and thus \eqref{eq:order2} will hold.  Finally, rewriting 
the last condition as: 
\begin{equation}
\frac{m_{(k+2)'}m_{(k+1)'}}{(1+m_{(k+1)'}^2\zeta_1)(1+m_{(k+2)'}^2\zeta_1)}<
\frac{\zeta_1^{\frac{k+1}{2}}}{\zeta_K^{\frac{k-1}{2}}}
\frac{2\, \textnormal{min}_j(\zeta_{j+1}-\zeta_j)^{k-1}}{(k+1)(\zeta_K-\zeta_1)^{k+1}}, 
\quad \text{ for all odd } k, \qquad k\leq 2K-3
\end{equation}
we obtain the second sufficient condition \eqref{eq:seccond}.  
\end{proof} 

As an example illustrating the global existence of our multipeakon solutions let us consider the case $K=2$ (i.e. $n=4$).  
\begin{example} Let $K=2$, and $ b_1(0)=10,\ b_2(0)=1,\ \zeta_1=0.3,\ \zeta_2=3,\ m_1=8,\ m_2=16,\ m_3=18,\ m_4=13$. It is easy to show that the condition in Theorem \ref{thm:global} is satisfied. Hence the order of $\{x_k, k = 1, 2, 3, 4\}$ will be preserved at all time and one can use the explicit formulae for the 4-peakon solution at all time, resulting in the following sequence of graphs (Figure \ref{fig_4peakon}).
\begin{figure}[h!]
  \centering
  \resizebox{1.1\textwidth}{!}{
  \includegraphics{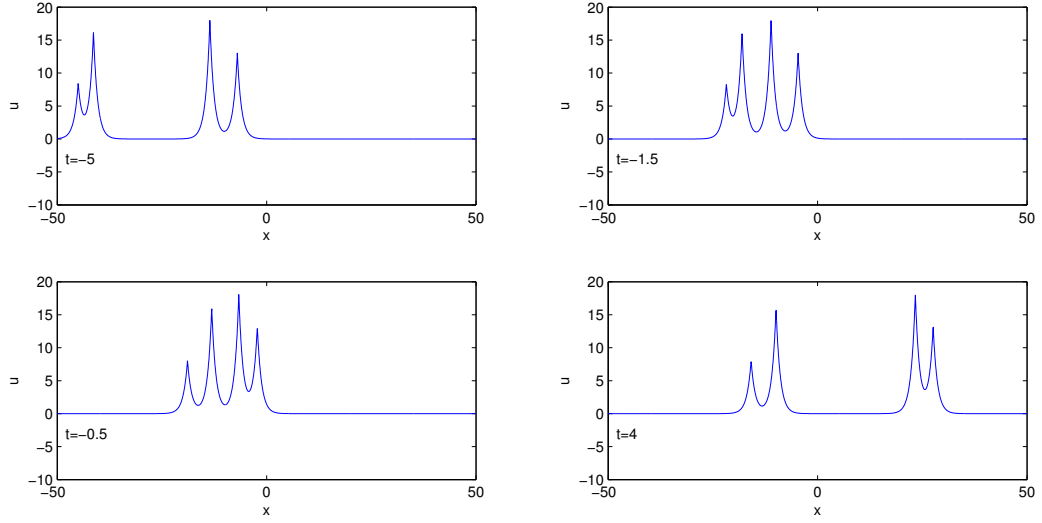}
  }
   \caption{Snapshots of $u(x,t)$ for $n=4$ at times $t=-5,\ -1.5,\ -0.5, \ 4$ in the case of $ b_1(0)=10,\ b_2(0)=1,\ \zeta_1=0.3,\ \zeta_2=3,\ m_1=8,\ m_2=16,\ m_3=18,\ m_4=13.$}
   \label{fig_4peakon}
\end{figure}
\end{example}

\subsection{Large time peakon asymptotics for $n=2K$}
In this short subsection we state the asymptotic behaviour of multipeakon 
solutions for large (positive and negative) time, thus implicitly assuming the 
global existence of solutions as guaranteed for example 
by imposing sufficient conditions of Theorem \ref{thm:global}.  

\begin{theorem}\label{thm:evenass} Suppose the masses $m_j$ satisfy the conditions of Theorem \ref{thm:global}.  Then the asymptotic position of a $k$-th (counting from the right) peakon as $t\rightarrow+\infty$ is given by 
\begin{subequations}
\begin{align}
&x_{k'}=\frac{2t}{\zeta_{\frac{k+1}{2}}}+
\ln\frac{b_{\frac{k+1}{2}}(0)\mathbf{e}_{[1,k-1]} \Delta^2_{[1,\frac{k-1}{2}],\{\frac{k+1}{2}\}}(\mathbf{\zeta})}{m_{k'} \Gamma_{[1,k], \{\frac{k+1}{2}\}}(\mathbf{e}; \mathbf{\zeta}) \mathbf{\zeta}^2_{[1,\frac{k-1}{2}]}}+\mathcal{O}(
e^{-\alpha_k t}), &\textrm{ for some positive } \alpha_k\, \textrm{ and odd } k, \\
&x_{k'}=\frac{2t}{\zeta_{\frac{k}{2}}}+
\ln\frac{b_{\frac{k}{2}}(0)\mathbf{e}_{[1,k-1]} \Delta^2_{[1,\frac{k}{2}-1],\{\frac{k}{2}\}}(\mathbf{\zeta})}{m_{k'} \Gamma_{[1,k-1], \{\frac{k}{2}\}}(\mathbf{e}; \mathbf{\zeta}) \mathbf{\zeta}^2_{[1,\frac{k}{2}-1]}\zeta_{\frac k2}}+\mathcal{O}(
e^{-\alpha_k t}), &\textrm{ for some positive } \alpha_k\, \textrm{ and even } k,\\
&x_{k'}-x_{(k+1)'}=\ln m_{(k+1)'}m_{k'} \zeta_{\frac{k+1}{2}}
+\mathcal{O}(e^{-\alpha_k t}), &\textrm{ for some positive } \alpha_k\, \textrm{ and odd } k.   
\end{align}
\end{subequations}

  Likewise, as $t\rightarrow-\infty$, using the notation of Theorem \ref{thm:Dklm-peakon}, the asymptotic position of the $k$-th peakon is given by
  \begin{subequations}
  \begin{align}
&x_{k'}=\frac{2t}{\zeta_{(\frac{k+1}{2})^*}}+
\ln\frac{b_{(\frac{k+1}{2})^*}(0)\mathbf{e}_{[1,k-1]} \Delta^2_{([1,\frac{k-1}{2}])^*,\{(\frac{k+1}{2})^*\}}(\mathbf{\zeta})}{m_{k'} \Gamma_{[1,k], \{(\frac{k+1}{2})^*\}}(\mathbf{e}; \mathbf{\zeta}) \mathbf{\zeta}^2_{[1,\frac{k-1}{2}]^*}}+\mathcal{O}(
e^{\beta_k t}), \qquad \textrm{ for some positive } \beta_k\, \textrm{ and odd } k, \\
&x_{k'}=\frac{2t}{\zeta_{(\frac{k}{2})^*}}+
\ln\frac{b_{(\frac{k}{2})^*}(0)\mathbf{e}_{[1,k-1]} \Delta^2_{[1,\frac{k}{2}-1]^*,\{(\frac{k}{2})^*\}}(\mathbf{\zeta})}{m_{k'} \Gamma_{[1,k-1], \{(\frac{k}{2})^*\}}(\mathbf{e}; \mathbf{\zeta}) \mathbf{\zeta}^2_{[1,\frac{k}{2}-1]^*}\zeta_{(\frac k2)^*}}+\mathcal{O}(
e^{\beta_k t}), \qquad \textrm{ for some positive } \beta_k\, \textrm{ and even } k,\\
&x_{k'}-x_{(k+1)'}=\ln m_{(k+1)'}m_{k'} \zeta_{(\frac{k+1}{2})^*}
+\mathcal{O}(e^{\beta_k t}), \qquad \textrm{ for some positive } \beta_k\, \textrm{ and odd } k.   
\end{align}
\end{subequations}
\end{theorem}

\begin{proof}
The proof is by a straightforward computation using the formulas 
for positions \eqref{eq:detinversexodd}, \eqref{eq:detinversexeven}, as 
well as asymptotic evaluations of determinants  \eqref{eq:Dklp+} and 
\eqref{eq:Dklp-}.  
\end{proof}

\begin{remark}
In the closing remark for this section we note that multipeakons 
of the mCH exhibit \textit{Toda-like sorting properties} of asymptotic speeds, 
a common occurrence among known to us peakon systems.  However, 
it is apparent that the multipeakons of the mCH show also features
known to occur in multi-component cases, for example in  
the Geng-Xue equation\cite{lundmark2014inverse}, a two-component modified Camassa-Holm equation \cite{chang-hu-szmigielski}, but also in the Novikov equation \cite{kardell-PhD, kardell-lundmark}, for which one observes an \textit{ asymptotic pairing of 
peakons}, undoubtedly forced by a shortage of eigenvalues whose total number is 
$K$,  versus $2K$ positions in need of asymptotic speeds.  This feature does not 
show up in the CH equation.  
\end{remark}

\section{Multipeakons for $n=2K+1$}\label{sec:oddpeakons}
The main source of difference with the even case is of course the presence of 
the positive shift $c$ which impacts the evaluations of 
the CSV determinants as illustrated by Theorem \ref{thm:detCSV}, in particular 
formula \eqref{eq:detCSV2}.  We will present the material 
in this section in a way parallel to the previous section on the even case.  
\subsection{Closed formulae for $n=2K+1$}

Again, we assume that $x_1(0)<x_2(0)<\cdots<x_{2K+1}(0)$ then this condition will hold at least in a small interval containing $t=0$.  Thus Theorem \ref{thm:inversex} gives us the following \textit{local existence} result.

\begin{theorem}\label{thm:peakon_odd}
Assuming the notation of Theorem \ref{thm:inversex}, the mCH equation \eqref{eq:m1CH} with the regularization 
of the singular term $u_x^2 m$ given by $\avg{u_x^2}m$ admits the multipeakon solution
\begin{equation}\label{eq:umultipeakoneven}
u(x,t)=\sum_{k=1}^{2K+1}m_{k'}(t)\exp(-|x-x_{k'}(t)|),
\end{equation}
where $x_{k'}$ are given by equations \eqref{eq:detinversexodd} and \eqref{eq:detinversexeven}, 
with the peakon spectral measure 
\begin{equation}\label{eq:peakon sm(odd)}
d\mu=\sum_{j=1}^{K} b_j(t) \delta_{\zeta_j}, 
\end{equation} 
$b_j(t)=b_j(0)e^{\frac{2t}{\zeta_j}}, \, 0<b_j(0)$, ordered eigenvalues $0<\zeta_1<\cdots<\zeta_K$ and $c(t)=c(0)>0$ in \eqref{eq:calD}.
\end{theorem}
\begin{proof} 
The time evolution of the spectral measure is the same 
as for the even case.  To see this as well as that $c$ is a constant 
we  recall that the Weyl function $W(z)$ is defined in \eqref{eq:defWeyl}, 
regardless of whether $n$ is even or odd, thus $W(z)$ 
undergoes the time evolution  obtained earlier in the proof of Theorem \ref{thm:peakon_even}, namely, 
\begin{equation*}
\dot W=\frac 2z W-\frac{2L}{z}, 
\end{equation*}
which, in turn, implies
$\dot b_j=\frac{2}{\zeta_j} b_j, 1\leq j\leq K$ as well as $\dot c=0$ by virtue of 
Corollary \ref{cor:spectrum}.  The rest of the proof is the same as for the even case.  
\end{proof}

We examine now the explicit formulas for the evaluation of 
CSV determinants presented in Theorem \ref{thm:detCSV} (see equation \eqref{eq:calD} for notation), with due care to two facts: $n=2K+1$ and $c>0$.  
The proof follows the same steps as in Theorem \ref{thm:Dklm-peakon} and 
we omit it.  

\begin{theorem} \label{thm:Dklp-peakon(odd)}  Let $n=2K+1, \, 1\leq k\leq 2K+1, \, 0\leq l\leq K+1, \, 0\leq p, \,  p+l-1\leq k-l, $ and let the peakon spectral measure be given by \eqref{eq:peakon sm(odd)} and a shift $c>0$.  
Then 
\begin{enumerate} 
\item 
\begin{subequations}
\begin{align} 
&\D_k^{(l,p)}=\abs{\Delta_{[1,k]}(\mathbf{e})}\sum_{I\in\binom{[1,K]}{l} }
\frac{\Delta^2_I(\mathbf{\zeta})\mathbf{b}_{I} \mathbf{\zeta}^p_I}{\Gamma_{[1,k],I}(\mathbf{e};\mathbf{\zeta})} \quad \text{ if} \quad p+l-1<k-l, \quad k\leq 2K+1; \label{eq:Dklpodd1}\\
&\D_k^{(l,p)}=\abs{\Delta_{[1,k]}(\mathbf{e})}\Big(\sum_{I\in\binom{[1,K]}{l} }
\frac{\Delta^2_I(\mathbf{\zeta})\mathbf{b}_{I} \mathbf{\zeta}^p_I}{\Gamma_{[1,k],I}(\mathbf{e};\mathbf{\zeta})} +c\sum_{I\in\binom{[1,K]}{l-1} }
\frac{\Delta^2_I(\mathbf{\zeta})\mathbf{b}_{I} \mathbf{\zeta}^p_I}{\Gamma_{[1,k],I}(\mathbf{e};\mathbf{\zeta})}\Big)
\quad \text{ if} \quad p+l-1=k-l, \quad k\leq 2K+1; \label{eq:Dklpodd2}
\end{align} 
\end{subequations}
with the proviso that the first term inside the bracket is set to zero if $l=K+1$, which only happens when 
$k=2K+1, p=0$.  
\item in the asymptotic region $t\rightarrow + \infty$
\begin{subequations}
\begin{align}
 &\D_k^{(l,p)}=\abs{\Delta_{[1,k]}(\mathbf{e})}
\frac{\Delta^2_{[1,l]} (\mathbf{\zeta})\mathbf{b}_{[1,l]} \mathbf{\zeta}^p_{[1,l]}}{\Gamma_{[1,k],[1,l]}(\mathbf{e};\mathbf{\zeta})}\Big[1+\mathcal{O}(e^{-\alpha t})\Big], \quad 0< \alpha, \quad \text{ if } \quad 0\leq l\leq K; \label{eq:Dklp+odd}\\
 &\D_{2K+1}^{(K+1,0)}=c\abs{\Delta_{[1,2K+1]}(\mathbf{e})}
\frac{\Delta^2_{[1,K]} (\mathbf{\zeta})\mathbf{b}_{[1,K]} }{\Gamma_{[1,2K+1],[1,K]}(\mathbf{e};\mathbf{\zeta})}, \quad  \qquad \text{ if } \quad k=2K+1, l=K+1, p=0.  \label{eq:Dklp+odd-c}
\end{align}
\end{subequations}
\item in the asymptotic region $t \rightarrow -\infty$
\begin{subequations}
\begin{align}
&\D_k^{(l,p)}=\abs{\Delta_{[1,k]}(\mathbf{e})}
\frac{\Delta^2_{[1,l]^*} (\mathbf{\zeta})\mathbf{b}_{[1,l]^*} \mathbf{\zeta}^p_{[1,l]^*}}{\Gamma_{[1,k],[1,l]^*}(\mathbf{e};\mathbf{\zeta})}\Big[1+\mathcal{O}(e^{\beta t})\Big]
, \quad 0<\beta, \quad \text{ if} \quad p+l-1<k-l, \quad k\leq 2K+1; \label{eq:Dklp-odd1}\\
&\D_k^{(l,p)}=c\abs{\Delta_{[1,k]}(\mathbf{e})}
\frac{\Delta^2_{[1,l-1]^*} (\mathbf{\zeta})\mathbf{b}_{[1,l-1]^*} \mathbf{\zeta}^p_{[1,l-1]^*}}{\Gamma_{[1,k],[1,l-1]^*}(\mathbf{e};\mathbf{\zeta})}\Big[1+\mathcal{O}(e^{\beta t})\Big]
, \quad 0<\beta, \quad  \quad \notag\\
&\hspace{7cm} \text{ if}  \qquad p+l-1=k-l,\quad  k<2K+1;\label{eq:Dklp-odd2}\\
&\D_{2K+1}^{(K+1,0)}=c\abs{\Delta_{[1,2K+1]}(\mathbf{e})}
\frac{\Delta^2_{[1,K]} (\mathbf{\zeta})\mathbf{b}_{[1,K]} }{\Gamma_{[1,2K+1],[1,K]}(\mathbf{e};\mathbf{\zeta})}, \quad  \qquad \text{ if } \quad k=2K+1, l=K+1, p=0, \label{eq:Dklp-odd3}
\end{align}
\end{subequations}
where, as before, $[1,l]^*=[l^*=K-l+1,1^*=K]$.  
\end{enumerate} 
\end{theorem} 
 For the future use, namely in the forthcoming proof of Theorem \ref{thm:global(odd)}, we will formulate an elementary 
 corollary aimed at comparing formulae with $c>$ and $c=0$.  
 For the duration of this corollary we explicitly display the dependence on $c$.  
 \begin{corollary} \label{cor:Dklp-c-dependence}
 Let $n=2K+1,\,  1\leq k\leq 2K+1, \,\, 0\leq l\leq K+1, \, 0\leq p, \,  p+l-1\leq k-l$,  and let the peakon spectral measure be given by \eqref{eq:peakon sm(odd)} and a shift $c>0$.  
Then 
\begin{subequations}
\begin{align} 
&\D_k^{(l,p)}(c)=\D_{k}^{(l,p)}(0), \quad &\text{ if} \quad p+l-1<k-l, \quad k\leq 2K+1; \label{eq:Dklpodd1-c-dep}\\
&\D_k^{(l,p)}(c)=\D_{k}^{(l,p)}(0)+c\D_k^{(l-1,p)}(0), 
\quad &\text{ if} \quad p+l-1=k-l, \quad k\leq 2K+1; \label{eq:Dklpodd2-c-dep}
\end{align} 
\end{subequations}
with the convention that the first term in \eqref{eq:Dklpodd2-c-dep} is set to zero if $l=K+1, k=2K+1,p=0$.      
\end{corollary} 
 \begin{proof} 
 It suffices to compare formulas \eqref{eq:Dklpodd1}
 and \eqref{eq:Dklpodd2} with \eqref{eq:Dklp}, 
 for $k\leq 2K$, while the case $k=2K+1$ can be 
 directly obtained from Theorem \ref{thm:detCSV} and the definition 
 of $\D_k^{(l,p)}(c)$ (see equation \eqref{eq:calD}).  
 \end{proof}
Finally, by use of the formulas \eqref{eq:detinversexodd}, 
\eqref{eq:detinversexeven}, Theorems \ref{thm:peakon_odd} and 
\ref{thm:Dklp-peakon(odd)}, in particular the formulas \eqref{eq:Dklpodd1} and \eqref{eq:Dklpodd2}, we get exact formulae for (local) 1,3-peakon solutions with 
initial positions satisfying $x_1(0)<x_2(0)<\cdots <x_{2K+1}(0)$.
\begin{example}[1-peakon solution; $K=0$ (trivial, does not require inverse 
spectral machinery)]
  \[
  x_1=\ln\left(\frac{c}{m_1}\right).
  \]
\end{example}
We note that by shifting this example covers the $1$-peakon solution 
discussed in Theorem 6.1 in \cite{gui2013wave}.  More generally, by shifting 
we cover all peakon solutions discussed therein for which masses are 
taken to be identical (see Remark \ref{re:olver} in the Introduction).  

\begin{example}[3-peakon solution; $K=1$]
  \begin{align*}
     &x_1=\ln\left(\frac{b_1c}{\zeta_1m_1\left(b_1\zeta_1m_2^2m_3^2+c(1+\zeta_1m_2^2)(1+\zeta_1m_3^2)\right)}\right),&\\
     &x_2=\ln\left(\frac{b_1m_2}{b_1\zeta_1m_2^2m_3^2+c(1+\zeta_1m_2^2)(1+\zeta_1m_3^2)}\left(\frac{b_1m_3^2}{1+\zeta_1m_3^2}+c\right)\right),&\\
     &x_3=\ln\left(\frac{1}{m_3}\left(\frac{b_1m_3^2}{1+\zeta_1m_3^2}+c\right)\right).
  \end{align*}

\end{example}

\begin{example}[a general formula for the last position $x_{1'}=x_{2K+1}$]

Recalling that $\hat \mu_c$ denotes the shifted Stieltjes transform of the spectral 
measure $\mu$ (introduced in definition \ref{def:CSV}) and using \eqref{eq:detinversexodd} we obtain: 
  \begin{equation*}
x_{1'}=x_{2K+1}=\ln \frac{\hat \mu_c(e_1)}{m_{1'}}=\ln \frac{c+m_{2K+1}^2 \sum_{i=1}^K 
\frac{b_i}{1+m_{2K+1}^2 \zeta_i}}{m_{2K+1}}. 
\end{equation*}
\end{example}

\subsection{Global existence for $n=2K+1$}

 This section presents the main results regarding the global existence 
 of peakon solutions when $n=2K+1$.  
\begin{theorem}\label{thm:global(odd)}
  Given arbitrary spectral data $$\{b_j>0, \, 0<\zeta_1<\zeta_2<\cdots< \zeta_K,\,  
  c>0: 1\leq j \leq K \}, $$ suppose the masses $m_k$ satisfy
\begin{subequations}
\begin{align}
&\frac{1}{m_{(k+1)' }m_{k'}}<\frac{\zeta_1^{\frac{k+1}{2}}}{\zeta_K^{\frac{k-1}{2}}}\textrm{min}\{1, \hat \beta\}, \qquad &\text { for all odd }k, \qquad   1\leq k\leq 2K-1, \\
&\frac{1}{m_{(k+2)'}m_{(k+1)}}<\frac{\zeta_1^{\frac{k+1}{2}}}{\zeta_K^{\frac{k-1}{2}}}\textrm{min}\{1, \hat \beta_1\}, 
\qquad &\text{ for all odd } k, \qquad 1\leq k\leq 2K-1,  
\end{align}
\end{subequations}
where 
\begin{align*}
\hat \beta&=\begin{cases} \frac{2\zeta_K\, \textrm{min}_j(\zeta_{j+1}-\zeta_j)^{k-3}}{\zeta_1(k-1)(\zeta_K-\zeta_1)^{k-1}}\frac{(1+m^2_{(k)'}\zeta_1)(1+m^2_{(k+1)'}\zeta_1)}{m^2_{(k)'} m^2_{(k+1)'}},  \quad &\text{ for all odd } k, \qquad 3\leq k\leq 2K-1,  \\
+\infty, \qquad &\text{ for } k=1, \end{cases}\\\\
\hat \beta_1&=\frac{2\, \textrm{min}_j(\zeta_{j+1}-\zeta_j)^{k-1}}{(k+1)(\zeta_K-\zeta_1)^{k+1}}\frac{(1+m^2_{(k+1)'}\zeta_1)(1+m^2_{(k+2)'}\zeta_1)}{m^2_{(k+1)'} m^2_{(k+2)'}}.  
\end{align*}
Then the positions obtained from inverse formulas \eqref{eq:detinversexodd}, 
\eqref{eq:detinversexeven} are ordered $x_1<x_2<\cdots<x_{2K+1}$ and the multipeakon solutions \eqref{eq:umultipeakoneven} exist for arbitrary $t\in \R$.
\end{theorem}
\begin{proof}
The solutions described in Theorem \ref{thm:inversex} are valid peakon solutions as long as $x_1<x_2<\cdots<x_{2K+1}$ holds.  We write these conditions as:

\begin{subequations}
\begin{align}
x_{(k+1)'}<x_{k'}, && \text { for all odd }k, \qquad  1\leq k\leq 2K-1, \label{eq:order1-odd}\\
x_{(k+2)'}<x_{(k+1)'}, && \text{ for all odd }k, \qquad  1\leq k\leq 2K-1, \label{eq:order2-odd}
\end{align}
\end{subequations}

and use equations \eqref{eq:detinversexodd}, \eqref{eq:detinversexeven} to 
obtain equivalent conditions
\begin{subequations}
\begin{align}
\frac{1}{m_{(k+1)'}m_{k'}}<\frac{\D_{k+1}^{(\frac{k+1}{2},1)}(c) \D_{k-1}^{(\frac{k-1}{2},0)}(c)}{\D_{k+1}^{(\frac{k+1}{2},0)} (c)\D_{k-1}^{(\frac{k-1}{2},1)}(c)},&& \text { for all odd }k, \qquad   1\leq k\leq 2K-1, \label{eq:order1bis-odd} \\
\frac{1}{m_{(k+2)'}m_{(k+1)'}}<\frac{\D_{k+2}^{(\frac{k+1}{2},1)}(c) \D_{k}^{(\frac{k+1}{2},0)}(c)}{\D_{k+2}^{(\frac{k+3}{2},0)}(c) \D_{k}^{(\frac{k-1}{2},1)}(c)},&& \text{ for all odd }k, \qquad 1\leq k\leq 2K-1,   \label{eq:order2bis-odd} 
\end{align} 
\end{subequations}
 displaying the dependence on $c$ in anticipation of the use of Corollary \ref{cor:Dklp-c-dependence}.  
Note that equations \eqref{eq:Dklpodd1-c-dep} and \eqref{eq:Dklpodd2-c-dep} 
suggest writing 
\begin{equation*}
\begin{split}
\frac{\D_{k+1}^{(\frac{k+1}{2},1)}(c) \D_{k-1}^{(\frac{k-1}{2},0)}(c)}{\D_{k+1}^{(\frac{k+1}{2},0)}(c) \D_{k-1}^{(\frac{k-1}{2},1)}(c)}=\frac{ \big(\D_{k+1}^{(\frac{k+1}{2}, 1)}(0)+c\D_{k+1}^{(\frac{k-1}{2}, 1)}(0)\big) \D_{k-1}^{(\frac{k-1}{2},0)}(0)}{
\D_{k+1}^{(\frac{k+1}{2},0)}(0)\big(\D_{k-1}^{(\frac{k-1}{2},1)}(0)+c\D_{k-1}^{(\frac{k-3}{2},1)}(0)\big)}\\=\frac{ \D_{k+1}^{(\frac{k+1}{2}, 1)}(0)\D_{k-1}^{(\frac{k-1}{2},0)}(0)+c\D_{k+1}^{(\frac{k-1}{2}, 1)}(0)\D_{k-1}^{(\frac{k-1}{2},0)}(0)}{
\D_{k+1}^{(\frac{k+1}{2},0)}(0)\D_{k-1}^{(\frac{k-1}{2},1)}(0)+c\D_{k+1}^{(\frac{k+1}{2},0)}(0)\D_{k-1}^{(\frac{k-3}{2},1)}(0)}\stackrel{def}{=} \frac {\mathcal{A}_1+\mathcal{B}_1}{\mathcal{A}_2+\mathcal{B}_2}, 
\end{split}
\end{equation*}
with the proviso that $\mathcal{B}_2=0$ for $k=1$.  Examining the ratios $\frac{\mathcal{A}_1}{\mathcal{A}_2}, \frac{\mathcal{B}_1}{\mathcal{B}_2}$ we observe that they satisfy (uniform in $t$) bounds
\begin{align*}
&\frac{\mathcal{A}_1}{\mathcal{A}_2}>\frac{\zeta_1^{\frac{k+1}{2}}}{\zeta_K^{\frac{k-1}{2}}}\stackrel{def}{=}\alpha, \\
&\frac{\mathcal{B}_1}{\mathcal{B}_2}>\frac{\zeta_1^{\frac{k-1}{2}}}{\zeta_K^{\frac{k-3}{2}}}\frac{2\, \textrm{min}_j(\zeta_{j+1}-\zeta_j)^{k-3}}{(k-1)(\zeta_K-\zeta_1)^{k-1}}\frac{(1+m^2_{k'}\zeta_1)(1+m^2_{(k+1)'}\zeta_1)}{m^2_{k'} m^2_{(k+1)'}}\stackrel{def}{=}\beta, 
\end{align*}
by equations \eqref{eq:order1step3} and \eqref{eq:order2step4}, respectively, 
with the convention that $\beta=\infty$ for the special case $k=1$.  
Thus 
\begin{equation*}
\textrm{min}\{\alpha, \beta\}<\frac{\D_{k+1}^{(\frac{k+1}{2},1)}(c) \D_{k-1}^{(\frac{k-1}{2},0)}(c)}{\D_{k+1}^{(\frac{k+1}{2},0)}(c) \D_{k-1}^{(\frac{k-1}{2},1)}(c)}
\end{equation*}
holds uniformly in $t$ and 
if we impose 
\begin{equation*}
\frac{1}{m_{(k+1)'}m_{k'}}<\textrm{min}\{\alpha, \beta\}\qquad \text { for all odd }k, \qquad   1\leq k\leq 2K-1,  
\end{equation*} 
then equations \eqref{eq:order1-odd} will hold automatically.  

Now we turn to the second inequality, namely \eqref{eq:order2bis-odd}. Again, using 
Corollary \ref{cor:Dklp-c-dependence} we obtain 
\begin{equation} \label{eq:order2step3odd}
\frac{\D_{k+2}^{(\frac{k+1}{2},1)}(c) \D_{k}^{(\frac{k+1}{2},0)}(c)}{\D_{k+2}^{(\frac{k+3}{2},0)}(c) \D_{k}^{(\frac{k-1}{2},1)}(c)}>\frac {\mathcal{A}_1+\mathcal{B}_1}{\mathcal{A}_2+\mathcal{B}_2}, 
\end{equation} 
where, this time, 
\begin{align*}
&\frac{\mathcal{A}_1}{\mathcal{A}_2}>\frac{\zeta_1^{\frac{k+1}{2}}}{\zeta_K^{\frac{k-1}{2}}}=\alpha, \\
&\frac{\mathcal{B}_1}{\mathcal{B}_2}>\frac{\zeta_1^{\frac{k+1}{2}}}{\zeta_K^{\frac{k-1}{2}}}\frac{2\, \textrm{min}_j(\zeta_{j+1}-\zeta_j)^{k-1}}{(k+1)(\zeta_K-\zeta_1)^{k+1}}\frac{(1+m^2_{(k+1)'}\zeta_1)(1+m^2_{(k+2)'}\zeta_1)}{m^2_{(k+1)'} m^2_{(k+2)'}}\stackrel{def}{=}\beta_1, 
\end{align*}
and, 
\begin{equation*}
\textrm{min}\{\alpha, \beta_1\}<\frac{\D_{k+2}^{(\frac{k+1}{2},1)}(c) \D_{k}^{(\frac{k+1}{2},0)}(c)}{\D_{k+2}^{(\frac{k+3}{2},0)}(c) \D_{k}^{(\frac{k-1}{2},1)}(c)}
\end{equation*}
is satisfied.  
Thus inequality 
\begin{equation*}
\frac{1}{m_{(k+2)'}m_{(k+1)'}}<\textrm{min}\{\alpha, \beta_1\}, \qquad  \text{ for all odd } k, \qquad 1\leq k\leq 2K-1,   
\end{equation*}
implies \eqref{eq:order2bis-odd} and, consequently, \eqref{eq:order2-odd}, 
thereby completing the proof.  
 
\end{proof}

\begin{example} Let $K=1$, and $ b_1(0)=1,\ c=3,\ \zeta_1=5,\ m_1=3,\ m_2=2,\ m_3=2.2.$ Then the sufficient conditions in Theorem \ref{thm:global(odd)} are satisfied. Hence the order of $\{x_k, k = 1, 2, 3\}$ will be preserved at all time and one can use the explicit formulae for the 3-peakon solution at all time, resulting in the following sequence of graphs (Figure \ref{fig_3peakon}).
\end{example}
\begin{figure}[h!]
  \centering
  \resizebox{1.1\textwidth}{!}{
  \includegraphics{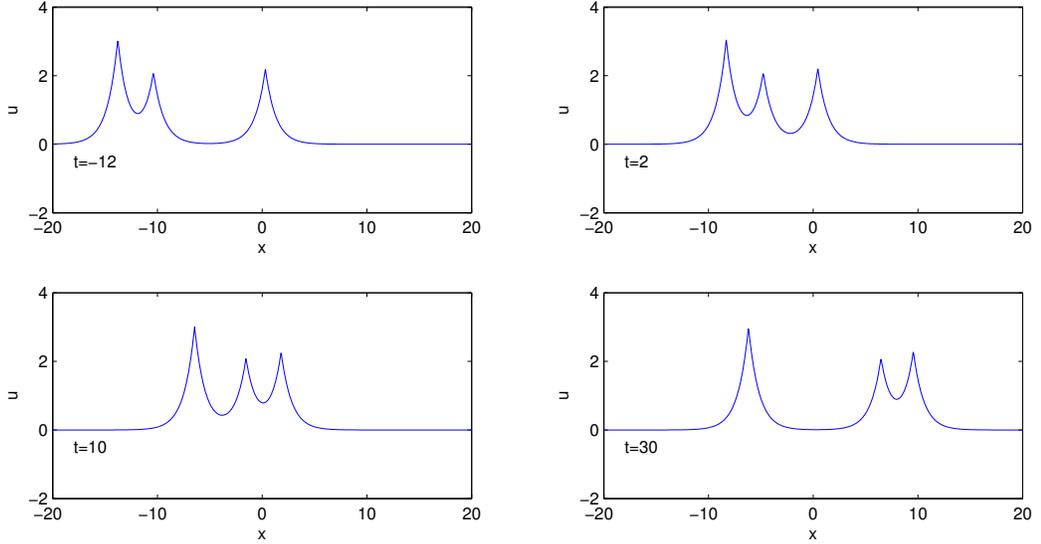}
  }
   \caption{Snapshots of $u(x,t)$ for $n=3$ at time $t=-12,\ 2,\ 10, \ 30$ in the case of $ b_1(0)=1,\ c=3,\ \zeta_1=5,\ m_1=3,\ m_2=2,\ m_3=2.2.$}
   \label{fig_3peakon}
\end{figure}

\subsection{Large time peakon asymptotics for $n=2K+1$}
We will investigate in this section the long time asymptotics 
of  global multipeakon solutions, guaranteed to exist by Theorem \ref{thm:global(odd)}.  
\begin{theorem}\label{thm:oddass} Suppose the masses $m_j$ satisfy the conditions of Theorem \ref{thm:global(odd)}.  Then the asymptotic position of a $k$-th (counting from the right) peakon as $t\rightarrow+\infty$ is given by 
\begin{subequations}
\begin{align}
&&x_{k'}&=\frac{2t}{\zeta_{\frac{k+1}{2}}}+
\ln\frac{b_{\frac{k+1}{2}}(0)\mathbf{e}_{[1,k-1]} \Delta^2_{[1,\frac{k-1}{2}],\{\frac{k+1}{2}\}}(\mathbf{\zeta})}{m_{k'} \Gamma_{[1,k], \{\frac{k+1}{2}\}}(\mathbf{e}; \mathbf{\zeta}) \mathbf{\zeta}^2_{[1,\frac{k-1}{2}]}}+\mathcal{O}(
e^{-\alpha_k t}), & \textrm{ for some positive } \alpha_k\, \textrm{ and odd } k\leq 2K-1; \\
&&x_{(2K+1)'}&=\ln \frac{c \mathbf{e}_{[1,2K]}}{m_{(2K+1)'}\mathbf{\zeta}_{[1,K]}^2} 
+\mathcal{O}(
e^{-\alpha t}), &\textrm{ for some positive } \alpha ; \\
&&x_{k'}&=\frac{2t}{\zeta_{\frac{k}{2}}}+
\ln\frac{b_{\frac{k}{2}}(0)\mathbf{e}_{[1,k-1]} \Delta^2_{[1,\frac{k}{2}-1],\{\frac{k}{2}\}}(\mathbf{\zeta})}{m_{k'} \Gamma_{[1,k-1], \{\frac{k}{2}\}}(\mathbf{e}; \mathbf{\zeta}) \mathbf{\zeta}^2_{[1,\frac{k}{2}-1]}\zeta_{\frac k2}}+\mathcal{O}(
e^{-\alpha_k t}), &\textrm{ for some positive } \alpha_k\, \textrm{ and even } k\leq 2K;\\
&&x_{k'}-x_{(k+1)'}&=\ln m_{(k+1)'}m_{k'} \zeta_{\frac{k+1}{2}}
+\mathcal{O}(e^{-\alpha_k t}),  &\textrm{ for some positive } \alpha_k\, \textrm{ and odd } k \leq 2K-1.   
\end{align}
\end{subequations}

  Likewise, as $t\rightarrow-\infty$, using the notation of Theorem \ref{thm:Dklm-peakon}, the asymptotic position of the $k$-th peakon is given by
  \begin{subequations}
  \begin{align}
  &&x_{k'}&=\frac{2t}{\zeta_{(\frac{k-1}{2})^*}}+
\ln\frac{b_{(\frac{k-1}{2})^*}(0)\mathbf{e}_{[1,k-1]} \Delta^2_{[1,\frac{k-1}{2}-1]^*,\{(\frac{k-1}{2})^*\}}(\mathbf{\zeta})}{m_{k'} \Gamma_{[1,k-1], \{(\frac{k-1}{2})^*\}}(\mathbf{e}; \mathbf{\zeta}) \mathbf{\zeta}^2_{[1,\frac{k-1}{2}-1]^*}\zeta_{(\frac {k-1}{2})^*}}+\mathcal{O}(
e^{\beta_k t}), &\textrm{ for positive } \beta_k \notag\\&&&&\textrm{ and odd  }1<k\leq 2K+1;\\
&&x_{1'}&=\ln\frac{c}{m_{1'}}+\mathcal{O}(
e^{\beta_k t}), & \textrm{ for positive } \beta_k;\\
&&x_{k'}&=\frac{2t}{\zeta_{(\frac{k}{2})^*}}+
\ln\frac{b_{(\frac{k}{2})^*}(0)\mathbf{e}_{[1,k-1]} \Delta^2_{([1,\frac{k}{2}-1])^*,\{(\frac{k}{2})^*\}}(\mathbf{\zeta})}{m_{k'} \Gamma_{[1,k], \{(\frac{k}{2})^*\}}(\mathbf{e}; \mathbf{\zeta}) \mathbf{\zeta}^2_{[1,\frac{k}{2}-1]^*}}+\mathcal{O}(
e^{\beta_k t}), &\textrm{ for positive } \beta_k\, \textrm{ and even }k; \\
&&x_{k'}-x_{(k+1)'}&=\ln m_{(k+1)'}m_{k'} \zeta_{(\frac{k}{2})^*}
+\mathcal{O}(e^{\beta_k t}), &\textrm{ for positive } \beta_k\, \textrm{ and even }k.   
\end{align}
\end{subequations}
\end{theorem}

\begin{proof}
The proof is by a straightforward, but tedious, computation using the formulas 
for positions \eqref{eq:detinversexodd}, \eqref{eq:detinversexeven}, as 
well as asymptotic evaluations of determinants  presented in Theorem \ref{thm:Dklp-peakon(odd)}.  
\end{proof}
\begin{remark} 
The Toda-like sorting property can also be observed in this case by examining 
more closely the asymptotic formulae but the pairing mechanism is subtly different.  
We point out that the constant $c$ is a surrogate of an additional eigenvalue $\zeta_{K+1}=\infty$, which results in the formal asymptotic speed $0$.  
Thus for large positive times the first particle counting from the left comes to a halt, while 
the remaining $2K$ peakons form pairs, sharing the remaining $K$ eigenvalues.  
By contrast, for large, negative times, the first particle counting from the right comes to a halt, 
while the remaining peakons form pairs.  This, somewhat intricate, breaking 
of symmetry is responsible for noticeable asymmetry in the indexing 
of positions seen when one compares asymptotic formulas for $n=2K$ with 
$n=2K+1$.  
\end{remark}
We would like to finish this with one application of asymptotic formulas, valid 
for any $n$, namely we will compute 
the Sobolev $H^1$ norm of $u$ which, by a result of \cite{chang-szmigielski-short1mCH}, is time invariant.  
\begin{corollary} \label{cor:sH1norm} Suppose 
masses satisfy conditions guaranteeing the global existence of 
solutions.  Then 
\begin{equation}\label{eq:H1norm}
||u||^2_{H^1}=2\sum_{j=1}^n m_j^2 +4 \sum_{j=1}^K \frac{1}{\zeta_j}. 
\end{equation}
\end{corollary}
\begin{proof} 
First, as proven in \cite{chang-szmigielski-short1mCH}, 
$||u||^2_{H^1}=\sum_{j=1}^n 2m_ju(x_j)=2\sum_{j=1}^n m_j^2+
4\sum_{i<j} m_i m_j e^{x_i-x_j}$, where we 
used the ordering condition $x_i<x_{i+1}$.  Since 
$||u||^2_{H^1}$ is constant we can compute its value 
using asymptotic formulas.  Thus from the asymptotic formulas in 
Theorem \ref{thm:oddass} (or \ref{thm:evenass} in the even case)
we see that the only contribution to the last term above will come 
from pairs sharing the same asymptotic speeds.  In other words, 
\begin{equation*}
||u||^2_{H^1}=\sum_{j=1}^n 2m_ju(x_j)=2\sum_{j=1}^n m_j^2+
\lim_{t\rightarrow +\infty} 
4\sum_{i=1}^K m_{2i}m_{2i+1} e^{x_{2i}-x_{2i+1}}=2\sum_{j=1}^n m_j^2+
4\sum_{i=1}^K\frac{1}{\zeta_i}, 
\end{equation*}
again, by asymptotic formulas of Theorem \ref{thm:oddass}, \ref{thm:evenass}, 
respectively.  
\end{proof} 

\section{Acknowledgements} 
The authors thank H. Lundmark for an interesting discussion 
regarding the asymptotic behaviour of peakon solutions.    

\begin{appendix} 
\section{Lax pair for the mCH peakon ODEs}\label{lax_mch}
 
  Our technique of solving the peakon ODEs (\ref{mCH_ode}) hinges on the following steps: 
  \begin{enumerate} 
  \item associate a Lax pair to the differential equation in question; 
  \item formulate the boundary value problem compatible with the Lax pair; 
  \item define the spectral data and its time evolution; 
  \item solve the inverse problem of reconstructing 
  the $x$ component of the Lax pair; 
  \end{enumerate} 
  One of the essential challenges of this program is to construct a well defined distribution Lax pair, i.e. a distribution version of \eqref{eq:xtLax}, which is ordinarily 
  given in the smooth sector of the equation.  
   The transition from the smooth sector to the distribution sector 
   is not canonical and this appendix addresses the main steps of our construction of 
   the correct distribution Lax pair used in this paper.  
  \begin{remark}
      In fact, we started our search for a distribution Lax pair suitable for \eqref{mch_ode_GLQ}.  We were, however, led to 
      a different definition of distribution solutions to \eqref{eq:m1CH} than in \cite{gui2013wave} or \cite{qiao2012integrable}.  Even though we do not 
      have a result that would exclude \eqref{mch_ode_GLQ} as coming from a suitably defined distribution Lax pair using some other way of 
      defining the products of distributions appearing in the Lax pair we can state this: within the class 
      of possible distribution Lax pairs which we will sharply define below  
      no such a pair exists.   \end{remark}

{\bf Notations}:
\begin{itemize} 

\item $\Omega_k$: the region $x_k(t)<x<x_{k+1}(t)$, where $x_k$ are smooth functions such that $-\infty=x_0(t)<x_1(t)<\cdots<x_{n}(t)<x_{n+1}(t)=+\infty$.

 \item  $PC^\infty$: the function space consisting of all the piecewise smooth functions $f(x,t)$ such that the restriction of $f$ to each region $\Omega_k$ is a smooth function $f_k(x,t)$ defined on an open neighbourhood of $\Omega_k$. Actually, for each fixed $t$, $f(x,t)$ defines a regular distribution $T_f(t)$ in the class of $\mathcal{D}'(R)$ (for simplicity we will write $f$ instead of $T_f$ ). Note that the value of $f(x,t)$ on $x_k(t)$ does not need to be defined.

\item   $f_x(x_k-,t)$: the left limit of the function $f(x,t)$ at every point $x_k$, $f_x(x_k+,t)$: the right limit of the function $f(x,t)$ at every point $x_k$.

 \item $\jump{f}(x_k,t)$: the jump between $f_x(x_k-,t)$ and $f_x(x_k+,t)$, i.e.
\[
\jump{f}(x_k,t)=f(x_k+,t)-f(x_k-,t).
\]

 \item   $\avg{f}(x_k,t)$: the arithmetic average of $f_x(x_k-,t)$ and $f_x(x_k+,t)$, i.e.
\[\avg{f}(x_k,t)=\frac{f(x_k+,t)+f(x_k-,t)}{2}.\]

\item   $f_x, f_t$: the ordinary (classical) partial derivative with respect to $x, t$ respectively.


  \item $D_xf$: the distributional derivative with respect to $x$.

\item   $D_tf$: the distributional limit $D_tf(t)=\lim_{a\rightarrow0}\frac{f(t+a)-f(t)}{a},.$
\item we will suppress the $t$-dependence throughout the remainder of this 
Appendix; thus $[f](x_k)$ will denote $[f](x_k,t)$ etc. 
\end{itemize} 
  Then the following identities follow from elementary 
  distributional calculus
\[
D_xf=f_x+\sum_{k=1}^{n}\jump{f}(x_k)\delta_{x_k}.
\]
\[
D_tf=f_t-\sum_{k=1}^{n}\dot x_{k}\jump{f}(x_k)\delta_{x_k},
\]
where $\dot x_{k}=\frac{dx_k}{dt}$.

Moreover, we also have:
\begin{equation}\label{eq:identities}
\begin{aligned}
  &\jump{fg}=\avg{f}\jump{g}+\jump{f} \avg{g},\qquad \avg{fg}=\avg{f}\avg{g}+\frac{1}{4}\jump{f} \jump{g},\\
  &\frac{d}{dt}\jump{f} (x_k)=\jump{f_x}(x_k)\dot x_k +\jump{f_t}(x_k),\\
  &\frac{d}{dt}\avg{f}(x_k)=\avg{ f_x}(x_k)\dot x_k +\avg{f_t}(x_k), 
\end{aligned}
\end{equation}
for any $f,g\in PC^\infty$.

It is easy to see that the peakon solution $u(x,t)$ and the corresponding functions $\Psi_1,\ \Psi_2$ belong to the piecewise smooth class $PC^\infty$. Indeed $u, u_x,\Psi_1,\ \Psi_2$ are smooth functions in $x_k<x<x_{k+1}$. However, $u$ is continuous throughout $\R$; by contrast $u_x, \Psi_1, \Psi_2$ have a jump at each $x_{k}$.

Let us now set $\Psi=(\Psi_1,\Psi_2)^T$, and let us consider an overdetermined system
\begin{equation}\label{eq:DLax-pair}
D_x\Psi=\frac{1}{2}\hat L\Psi,\qquad D_t\Psi=\frac{1}{2}\hat A\Psi,
\end{equation}
where
\begin{eqnarray}
  &&\hat L=L+2\lambda\left(\sum_{k=1}^nm_{k}\delta_{x_{k}}\right)M,\label{lax_peakon1}\\
  &&\hat A=A-2\lambda\left(\sum_{k=1}^nm_{k}Q(x_{k})\delta_{x_{k}}\right)M\label{lax_peakon2}
\end{eqnarray}
with
\begin{eqnarray*}
  &&L=\left(
\begin{array}{cc}
     -1     & 0  \\
 0&    1 \\
\end{array}
\right),
\qquad M=\left(
\begin{array}{cc}
     0     & 1  \\
 -1&    0 \\
\end{array}
\right),
 \qquad A=\left( \begin{array}{cc}
4\lambda^{-2}+Q              &  -2\lambda^{-1}(u-u_x)\\
2\lambda^{-1}(u+u_x) & -Q
\end{array} \right) 
\end{eqnarray*}
and $Q=u^2-u_x^2$. Note that in view of (\ref{lax_peakon1}) the $x$-member of the Lax equation \eqref{eq:DLax-pair} involves multiplying $M\Psi=(\Psi_2,-\Psi_1)$ by $\delta_{x_{k}}$. Thus we have to assign some values to $\Psi_1,\Psi_2$ at $x_k$. Likewise, for the $t$-Lax equation (\ref{lax_peakon2}) to be defined as a distribution equation, $u_x^2M\Psi=(u_x^2\Psi_2,-u_x^2\Psi_1)$ needs to be a multiplier of $\delta_{x_{k}}$. Thus the values of $u_x^2(x_{k})$ need to be assigned as well. 
Henceforth, we will refer to these assignments as \textit{regularizations}.   
The compatibility condition $(D_xD_t-D_tD_x)\Psi=0$ is a geometric 
condition (the zero curvature condition), and can be written as 
\begin{equation*}
\big(D_x(\hat A)-D_t(\hat L)+\frac12[\hat A, \hat L]\big)\Psi=0, 
\end{equation*}
whose invariance includes the transformations $\Psi \rightarrow \Psi_R=R \Psi, 
R\in GL(2,\R)$.  These transformations leave the singular support of $m$ invariant, 
and we require that the assignment of values to $\Psi$ on the singular support 
respects that symmetry.  Thus 
we postulate that for every $x_k$ 
\begin{equation*}
\Psi_R(x_k)=R\Psi(x_k),  \qquad R\in GL(2,\R). 
\end{equation*}
Furthermore we consider local regularizations, depending only on 
the right and left hand limits at the points of singular support.  
In summary we consider regularizations of the form: 
\begin{equation}\label{eq:localreg} 
\Psi(x_k)=\alpha \jump{\Psi}(x_k)+\beta\avg{\Psi}(x_k), \qquad \alpha, \beta \in GL(2,\R), 
\end{equation}
which lead, under the invariance assumption, to the condition: 
\begin{equation*}
\Psi_R(x_k)=\alpha \jump{\Psi_R}(x_k)+\beta\avg{\Psi_R}(x_k)=R\big(
\alpha \jump{\Psi}(x_k)+\beta\avg{\Psi}(x_k)\big)
\end{equation*}
valid for every $R\in GL(2,\R)$ and resulting in the intertwining conditions
\begin{equation}
{\mathbf{\alpha}}R=R\alpha, \quad \beta R=R\beta, \qquad \forall R\in GL(2,\R).   
\end{equation}
Consequently, by Schur's Lemma $\alpha$ and $\beta$ are scalar matrices.  This 
motivates the next definition.  
\begin{definition} \label{def:invariantreg}
An invariant regularization of 
the Lax pair \eqref{eq:DLax-pair} is given by specifying 
the values of $\alpha, \beta \in \R$ and $Q(x_k)=(u^2 -u_x^2)(x_k)$ in the formulas
below
\begin{equation*}
\begin{aligned}
\Psi(x)\delta _{x_k}&=\Psi(x_k) \delta_{x_k}, \\
\Psi(x_k)&=\alpha \jump{\Psi}(x_k)+\beta\avg{\Psi}(x_k),  \\
Q(x)\delta_{x_k}&=Q(x_k) \delta_{x_k}. 
\end{aligned}
\end{equation*}
\end{definition}

\begin{theorem}\label{thm:invreg}
Let $m$ be the discrete measure associated to $u$ defined by \eqref{eq:peakonansatz}.
Given an invariant regularization in the sense of \ref{def:invariantreg} the 
distributional Lax pair \eqref{eq:DLax-pair} is compatible, i.e. $D_tD_x\Psi=D_xD_t\Psi$,  if and only if 
the following conditions hold: 
\begin{subequations} 
\begin{align} 
\beta^2&=4 \alpha^2, \label{eq:betaalpha}\\
\beta &=1,  \label{eq:beta} \\
Q(x_k)&=\avg{Q}(x_k), \label{eq:Qvalue}\\
\dot m_k&=0, \label{eq:dotmk}\\
\dot x_k&=Q(x_k). \label{eq:dotxk}
\end{align}
\end{subequations}

\end{theorem}
\begin{proof}
The proof proceeds in a similar way to Theorem B.1 in \cite{hls} (also see \cite{chang-hu-szmigielski}). We highlight the critical steps 
of the proof.  First, we observe that since we are interested only 
in the behaviour of Lax pairs around the singular points $x_k$ we can 
\textit{localize} our computations to be carried out only locally 
on some open neighbourhoods of these points.  Moreover, these computations 
look identical, regardless of the index $k$.  In other words, 
without loss of generality we can assume $u(x)=m_1e^{-\abs{x-x_1}}$ for the sake 
of the computation, thus using $n=1$, and then in the final step of the proof 
pass to a general $n$.  
With this simplification in mind, assuming invariant 
regularization \ref{def:invariantreg}, we write equation \eqref{eq:DLax-pair} as
\begin{align*}
D_x\Psi&=\frac12 L \Psi + \lambda m_1 M \Psi(x_1) \delta_{x_1}, \\
D_t\Psi &=\frac12 A \Psi -\lambda m_1 Q(x_1)M\Psi(x_1) \delta_{x_1}. 
\end{align*}
In particular, the first equation implies
\begin{equation}\label{eq:jumpPsi}
\jump{\Psi}(x_1)=\lambda m_1 M \Psi(x_1). 
\end{equation} 
The computation of the distribution compatibility condition $D_x D_t \Psi=
D_t D_x \Psi$ produces a distribution condition which can be split 
into the regular and singular parts.  The regular part is just the compatibility 
condition one gets in the smooth sector of the equation and we omit that.  
The singular part takes the form: 
\begin{equation*}
\begin{split}
&[\frac12 A \Psi](x_1)\delta_{x_1}-\lambda m_1 Q(x_1)M\Psi(x_1) \delta_{x_1}'
=\\
&-\frac{\lambda}{2} m_1 Q(x_1)\, L M\Psi(x_1) \delta_{x_1}+\lambda\big(\dot m_1M\Psi(x_1)+
\ m_1 M \dot\Psi(x_1) \big)\delta_{x_1}-\lambda m_1 M \Psi(x_1)\dot x_1 \delta_{x_1}'. 
\end{split}
\end{equation*}
The coefficients at $\delta_{x_1}'$ imply equation \eqref{eq:dotxk}, while 
the coefficients at $\delta_{x_1}$ give the condition: 
\begin{equation}\label{eq:intermZCC}
[\frac12 A \Psi](x_1)=-\frac{\lambda}{2} m_1 Q(x_1)\, L M\Psi(x_1)+\lambda\big(\dot m_1M\Psi(x_1)+
\ m_1 M \dot\Psi(x_1) \big).  
\end{equation}
Since the value of $\Psi(x_1)$ is determined uniquely once 
the coefficients $\alpha$ and $\beta$ are chosen and the 
values of $Q(x_k)$ are assigned  (fixing a regularization) we can compute 
the term $\dot \Psi(x_1)$ appearing in \eqref{eq:intermZCC} with the help of equations \eqref{eq:identities}, \eqref{eq:jumpPsi}, the definition \ref{def:invariantreg},  and \eqref{eq:dotxk}.  After several intermediate elementary steps we obtain: 
\begin{equation}\label{eq:dotPsi}
\dot \Psi(x_1)=\big\{\avg{\frac A2}(x_1)+\frac\alpha\beta\jump{\frac A2}(x_1)+ 
\lambda m_1 (\frac \beta4-\frac{\alpha^2}{\beta})\jump{\frac A2}(x_1) M+\dot x_1 \frac L2\big\}\Psi(x_1). 
\end{equation}
Likewise, we can express the right hand side of \eqref{eq:intermZCC} by using 
\eqref{eq:identities}, \eqref{eq:jumpPsi} and \ref{def:invariantreg}. Again, 
after some straightforward computations we obtain: 
\begin{equation} \label{eq:LHS}
\jump{\frac A2 \Psi}(x_1)=\big\{\avg{\frac A2}(x_1) \lambda m_1 M+ 
\jump{\frac A2}(x_1) \frac{1-\alpha \lambda m_1 M}{\beta} \big \} \Psi(x_1),  
\end{equation}
which finally gives us the compatibility condition we have set out to obtain: 
\begin{equation} \label{eq:ZCC}
\begin{split}
&\lambda m_1\avg{\frac A2}(x_1)  M+ 
\jump{\frac A2}(x_1) \frac{1-\alpha \lambda m_1 M}{\beta}\\
&=-\lambda m_1 Q(x_1)\frac12 L M+\lambda\dot m_1M+
\lambda \ m_1 M \big\{\avg{\frac A2}(x_1)+\frac\alpha\beta\jump{\frac A2}(x_1)+ 
\lambda m_1 (\frac \beta4-\frac{\alpha^2}{\beta})\jump{\frac A2}(x_1) M+Q(x_1) \frac L2\big\}.  
\end{split}
\end{equation} 
We now summarize the content of \eqref{eq:ZCC}, broken down 
according to powers of $\lambda$, omitting conditions identically satisfied, 
\begin{enumerate}
\item $\lambda^{-1}: \quad \frac{\jump{u_x}(x_1)}{\beta}=-2m_1$
\item $\lambda^{1}: \quad \dot m_1=m_1\big(Q(x_1)-\avg{Q}(x_1)\big), \quad 
\dot m_1=-m_1\big(Q(x_1)-\avg{Q}(x_1)\big)$, 
\item $\lambda^2: \quad \frac \beta4-\frac{\alpha^2}{\beta}=0$,  
\end{enumerate} 
which imply claims \eqref{eq:beta}, \eqref{eq:Qvalue}, \eqref{eq:dotmk}, \eqref{eq:betaalpha} after restoring the number of singular points to $n$.  

\end{proof}

\begin{corollary}\label{cor:invreg}
There are only two invariant regularizations of the Lax pair \eqref{eq:xtLax} for the peakon problem of 
the mCH equation \eqref{eq:m1CH}: 
\begin{equation}
\Psi(x_k)=\Psi(x_k+), \qquad \textrm{or} \qquad \Psi(x_k)=\Psi(x_k-).  
\end{equation}
For either of the two regularizations $u_x^2(x_k)=\avg{u_x^2}(x_k)$ and in both cases the equations of motion read: 
\begin{equation}
\dot m_k=0, \qquad \dot x_k=u^2(x_k)-\avg{u_x^2}(x_k). 
\end{equation}
\end{corollary}
\begin{remark} 
In the body of the paper we use both regularizations to define the right and the 
left boundary value problems. 
\end{remark} 
\begin{remark} Observe that one does not need to specify the values 
of $u_x(x_k)$.  
\end{remark}

\end{appendix}

\def\cydot{\leavevmode\raise.4ex\hbox{.}}
  \def\cydot{\leavevmode\raise.4ex\hbox{.}}

\end{document}